\documentclass[acmsmall,screen]{acmart}
\usepackage{pgfplots}

\usepackage{tikz}
\usetikzlibrary{arrows}
\usetikzlibrary{cd} % commutative diagrams
\usetikzlibrary{decorations.pathmorphing}

\usepackage{mathpartir}

\usepackage{amsmath}
\usepackage{booktabs}
\usepackage{subcaption}
\usepackage{graphicx}
\usepackage[utf8]{inputenc}
\usepackage{multirow}
\usepackage{tabularx}

\usepackage{cmap}
\usepackage[T1]{fontenc}

\usepackage[scaled=0.75]{beramono}

\usepackage{nanoml}

\newcommand{\true}{\text{true}}
\newcommand{\false}{\text{false}}
\newcommand{\code}[1]{\mbox{\lstinline[language=nanoml,basicstyle=\small\ttfamily,columns=flexible,mathescape=true]^#1^}}

\newcommand{\scite}[1]{\footnote{\citet*{#1}}}

\setcopyright{rightsretained}
\acmPrice{}
\acmDOI{10.1145/3234599}
\acmYear{2018}
\copyrightyear{2018}
\acmJournal{PACMPL}
\acmVolume{2}
\acmNumber{ICFP}
\acmArticle{10}
\acmMonth{9}

\begin{document}

%% Title information
\title{Capturing the Future by Replaying the Past}
\subtitle{Functional Pearl}

%% Author information
%% Contents and number of authors suppressed with 'anonymous'.
%% Each author should be introduced by \author, followed by
%% \authornote (optional), \orcid (optional), \affiliation, and
%% \email.
%% An author may have multiple affiliations and/or emails; repeat the
%% appropriate command.
%% Many elements are not rendered, but should be provided for metadata
%% extraction tools.

%% Author with single affiliation.
\author{James Koppel}
%\authornote{with author1 note}          %% \authornote is optional;
                                        %% can be repeated if necessary
%\orcid{nnnn-nnnn-nnnn-nnnn}             %% \orcid is optional
\affiliation{
%  \position{Position1}
%  \department{Department1}              %% \department is recommended
  \institution{MIT}            %% \institution is required
%  \streetaddress{Street1 Address1}
  \city{Cambridge}
  \state{MA}
%  \postcode{Post-Code1}
  \country{USA}
}
\email{jkoppel@mit.edu}          %% \email is recommended

\author{Gabriel Scherer}
\orcid{0000-0003-1758-3938}              % \orcid is optional
\affiliation{
  \institution{INRIA}
 \country{France}
}
\email{gabriel.scherer@inria.fr}       %% \email is recommended

\author{Armando Solar-Lezama}
%\authornote{with author1 note}          %% \authornote is optional;
                                        %% can be repeated if necessary
%\orcid{nnnn-nnnn-nnnn-nnnn}             %% \orcid is optional
\affiliation{
%  \position{Position1}
%  \department{Department1}              %% \department is recommended
  \institution{MIT}            %% \institution is required
%  \streetaddress{Street1 Address1}
  \city{Cambridge}
  \state{MA}
%  \postcode{Post-Code1}
  \country{USA}
}
\email{asolar@csail.mit.edu}          %% \email is recommended

%% Paper note
%% The \thanks command may be used to create a "paper note" ---
%% similar to a title note or an author note, but not explicitly
%% associated with a particular element.  It will appear immediately
%% above the permission/copyright statement.
%\thanks{with paper note}                %% \thanks is optional
                                        %% can be repeated if necesary
                                        %% contents suppressed with 'anonymous'

\begin{abstract}
Delimited continuations are the mother of all monads! So goes the slogan inspired by Filinski's 1994 paper, which showed that delimited continuations can implement any monadic effect, letting the programmer use an effect as easily as if it was built into the language. It's a shame that not many languages have delimited continuations.

Luckily, exceptions and state are also the mother of all monads! In this Pearl, we show how to implement delimited continuations in terms of exceptions and state, a construction we call {\it thermometer continuations}. While traditional implementations of delimited continuations require some way of "capturing" an intermediate state of the computation, the insight of thermometer continuations is to reach this intermediate state by replaying the entire computation from the start, guiding it using a recording so that the same thing happens until the captured point.

Along the way, we explain delimited continuations and monadic reflection, show how the Filinski construction lets thermometer continuations express any monadic effect, share an elegant special-case for nondeterminism, and discuss why our construction is not prevented by theoretical results that exceptions and state cannot macro-express continuations.
\end{abstract}

%% 2012 ACM Computing Classification System (CSS) concepts
%% Generate at 'http://dl.acm.org/ccs/ccs.cfm'.


%% End of generated code

%% Keywords
%% comma separated list
  %% \keywords is optional

\maketitle

\section{Introduction}
In the days when mainstream languages have been adopting higher-order functions, advanced monadic effects like continuations and nondeterminism have held out as the province of the bourgeois programmer of obscure languages. Until now, that is.

Of course, there's a difference between effects which are built into a language and those that must be encoded. Mutable state is built-in to C, and so one can write \code{int x = 1; x += 1; int y = x + 1;}. Curry \scite{hanus1995curry} is nondeterministic, and so one can write \code{(3 ?\ 4) * (5 ?\ 6)}, which evaluates to all of $\{15,18,20,24\}$. This is called the \emph{direct style}. When an effect is not built into a language, the \emph{monadic}, or \emph{indirect}, style is needed. In the orthodox indirect style, after every use of an effect, the remainder of the program is wrapped in a lambda. For instance, the nondeterminism example could be rendered in Scala as \code{List(3,4).flatMap(x => List(5,6).flatMap(y => List(x * y)))}. Effects implemented in this way are called \emph{monadic}. "Do-notation," as seen in Haskell, makes this easier, but still inconvenient.

We show that, in any language with exceptions and state, you can implement any monadic effect in direct style. With our construction, you could implement a \code{?} operator in Scala such that the example \code{(3 ?\ 4) * (5 ?\ 6)} will run and return \code{List(15,18,20,24)}. Filinski showed how to do this in any language that has an effect called \emph{delimited continuations} or \emph{delimited control}.\scite{Filinski94} We first show how to implement delimited continuations in terms of exceptions and state, a construction we call \emph{thermometer continuations}, named for a thermometer-like visualization. Filinski's result does the rest. Continuations are rare, but exceptions are common. With thermometer continuations, you can get any effect in direct style in 9 of the TIOBE top 10 languages (all but C).\scite{tiobe10}

Here's what delimited continuations look like in cooking: Imagine a recipe for making chocolate nut bars. Soak the almonds in cold water. Rinse, and grind them with a mortar and pestle. Delimited continuations are like a step that references a sub-recipe. Repeat the last two steps, but this time with pistachios. Delimited continuations can perform arbitrary logic with these subprograms ("Do the next three steps once for each pan you'll be using"), and they can abort the present computation ("If using store-bought chocolate, ignore the previous four steps"). They are "delimited" in that they capture only a part of the program, unlike traditional continuations, where you could not capture the next three steps as a procedure without also capturing everything after them, including the part where you serve the treats to friends and then watch Game of Thrones. Implementing delimited continuations requires capturing the current state of the program, along with the rest of the computation up to a "delimited" point. It's like being able to rip out sections of the recipe and copy them, along with clones of whatever ingredients have been prepared prior to that section. This is a form of "time travel" that typically requires runtime support --- if the nuts had not yet been crushed at step N, and you captured a continuation at step N, when it's invoked, the nuts will suddenly be uncrushed again.

The insight of thermometer continuations is that every subcomputation is contained within the entire computation, and so there is an alternative to time travel: just repeat the entire recipe from the start! But this time, use the large pan for step 7. Because the computation contains delimited control (which can simulate any effect), it's not guaranteed to do the same thing when replayed. Thermometer continuations hence record the result of all effectful function calls so that they may be replayed in the next execution: the \emph{past} of one invocation becomes the \emph{future} of the next. Additionally, like a recipe step that overrides previous steps, or that asks you to let it bake for an hour, delimited continuations can abort or suspend the rest of the computation. This is implemented using exceptions.

% NOTE: I'm being a bit loose with language. They can't implement all other effects
% only monadic effects whose unit and bind functions are expressible as a purely functional term

This approach poses an obvious limitation: the replayed computation can't have any side effects, except for thermometer continuations. And replays are inefficient. Luckily, thermometer continuations can implement all other effects, and there are optimization techniques that make it less inefficient. Also, memoization --- a "benign effect" --- is an exception to the no-side-effects rule, and makes replays cheaper. The upshot is that our benchmarks in Section \ref{sec:benchmarks} show that thermometer continuations perform surprisingly well against other techniques for direct-style effects.

%"Give an intuition....." <-- clutter!
Here's what's in the rest of this Pearl: Our construction has an elegant special case for nondeterminism, presented in Section \ref{sec:nondet}, which also serves as a warm-up to full delimited control. In the following section, we give an intuition for how to generalize the nondeterminism construction with continuations. Section \ref{sec:dc} explains thermometer continuations. Section \ref{sec:monadic-reflection} explains Filinski's construction and how it combines with thermometer continuations to get arbitrary monadic effects. Section \ref{sec:optimization} discusses how to optimize the fusion of thermometer continuations with Filinski's construction, while Section \ref{sec:benchmarks} provides a few benchmarks showing that thermometer continuations are not entirely impractical. Finally, Section \ref{sec:compare-theoretical} discusses why our construction does not contradict a theoretical result that exceptions and state cannot simulate continuations. We also sketch a correctness proof of replay-based nondeterminism in Appendix \ref{sec:correctness}.

\section{Warm-Up: Replay-Based Nondeterminism}

\label{sec:nondet}

Nondeterminism is perhaps the first effect students learn which is not readily available in a traditional imperative language. This section presents replay-based nondeterminism, a useful specialization of thermometer continuations, and an introduction to its underlying ideas.

When writing the examples in this paper, we sought an impure language with built-in support for exceptions and state, and which has a simple syntax with good support for closures. We hence chose to present in SML. For simplicity, this paper will provide implementations that use global state and are hence not thread-safe. We assume they will not be used in multi-threaded programs.

Nondeterminism provides a choice operator \code{choose} such that \code{choose [x1, x2, ...]} may return any of the $x_i$. Its counterpart is a \code{withNondeterminism} operator which executes a block that uses \code{choose}, and returns the list of values resulting from all executions of the block.

\begin{minipage}{\textwidth}
\begin{nanoml}
withNondeterminism (fn () =>
  (choose [2,3,4]) * (choose [5,6]))
(* val it = [10,12,15,18,20,24] : int list *)
\end{nanoml}
\end{minipage}

\noindent In this example, there are six resulting possible values, yet the body returns one value. It hence must run six times. The replay-based implementation of nondeterminism does exactly this: in the first run, the two calls to \code{choose} return $2$ and $5$, then $2$ and $6$ in the second, etc. In doing so, the program behaves as if the first call to \code{choose} was run once but returned thrice. We'll soon show exactly how this is done. But first, let us connect our approach to the one most familiar to Haskell programmers: achieving nondeterminism through monads.

In SML, a monad is any module which implements the following signature (and satisfies the monad laws):

\begin{minipage}{\textwidth}
\begin{nanoml}
signature MONAD = sig
    type 'a m
    val return : 'a -> 'a m
    val bind :  'a m -> ('a -> 'b m) -> 'b m
end;
\end{nanoml}
\end{minipage}

\noindent Here is the implementation of the list monad in SML:

\label{sec:contains-list-monad}

\begin{minipage}{\textwidth}
\begin{nanoml}
structure ListMonad : MONAD = struct
  type 'a m = 'a list
  fun return x = [x]

  fun bind []      f = []
    | bind (x::xs) f = f x @ bind xs f
end;
\end{nanoml}
\end{minipage}

\noindent The ListMonad lets us rewrite the above example in monadic style. In the direct style, \code{choose [2,3,4]} would return thrice, causing the rest of the code to run thrice. Comparatively, in the monadic style, the rest of the computation is passed as a function to \code{ListMonad.bind}, which invokes it thrice.

\begin{nanoml}
- open ListMonad;

- bind [2,3,4] (fn x =>
  bind [5,6]   (fn y =>
    return (x * y)))
(* val it = [10,12,15,18,20,24] : int list *)
\end{nanoml}

\noindent Let's look at how the monadic version is constructed from the direct
one. From the perspective of the invocation \code{choose [2, 3, 4]},
the rest of the expression is like a function awaiting its result, which it must invoke thrice:

\begin{nanoml}
C = (fn \hole => \hole * choose [5, 6])
\end{nanoml}

\noindent This remaining computation is the \emph{continuation} of \code{choose [2,3,4]}. Each time \code{choose [2,3,4]} returns, it invokes this continuation. The monadic transformation \emph{captured} this continuation, explicitly turning it into a function. This transformation captures the continuation at compile time, but it can also be captured at runtime with the \code{call/cc} "call with current continuation" operator: if this first call to \code{choose} were replaced with \code{call/cc (fn k => ...)}, then \code{k} would be equivalent to $C$. So, the functions being passed to bind are exactly what would be obtained if the program were instead written in direct style and used \code{call/cc}.

This insight makes it possible to implement a direct-style \code{choose} operator. The big idea is that, once \code{call/cc} has captured that continuation \code{C} in \code{k}, it must invoke \code{k} thrice, with values $2$, $3$, $4$. This implementation is a little verbose in terms of \code{call/cc}, but we'll later see how delimited continuations make this example simpler than with \code{call/cc}-style continuations.

Like the monadic and \code{call/cc}-based implementations of
nondeterminism, replay-based nondeterminism invokes the continuation
\code{(TEXT_ONLY_HOLE => TEXT_ONLY_HOLE$\;$ * choose [5,6])} three times. Since the program \code{(choose [2,3,4] * choose [5,6])} is in direct style, and it cannot rely on a built-in language mechanism to capture the continuation, it does this by running the entire block multiple times, with some bookkeeping to coordinate the runs. We begin our explanation of replay-based nondeterminism with the simplest case: a variant of \code{choose} which takes only two arguments, and may only be used once.

\subsection{The Simple Case: Two-choice nondeterminism, used once}
\label{subsec:bool-once}

We begin by developing the simplified \code{choose2} operator. Calling \code{choose2 (x,y)} splits the execution into two paths, returning \code{x} in the first path and \code{y} in the second. For example:

\begin{nanoml}
 - (withNondeterminism2 (fn () => 3 * choose2 (5, 6)))
 ==>  [3 * 5, 3 * 6]
 ==>  [15, 18]
\end{nanoml}

\noindent This execution trace hints at an implementation in which \code{withNondetermism2} calls the block twice, and where \code{choose2} returns the first value in the first run, and the second value in the second run. \code{withNondeterminism2} uses a single bit of state to communicate to \code{choose2} whether it is being called in the first or second execution. As long as the block passed to \code{withNondeterminism2} is pure, with no effects other than the single use of \code{choose} (and hence no nested calls to \code{withNondeterminism2}), each execution will have the same state at the time \code{choose} is called.

\begin{minipage}{\textwidth}
\begin{nanoml}
val firstTime = ref false

(* choose2 : 'a * 'a -> 'a *)
fun choose2 (x1,x2) = if !firstTime then x1 else x2

(* withNondeterminism2 : (unit -> 'a) -> 'a list *)
fun withNondeterminism2 f = [(firstTime := true ; f ()),
                             (firstTime := false; f ())]

- withNondeterminism2 (fn () => 3 * choose2 (5,6))
(* val it = [15,18] : int list *)
\end{nanoml}
\end{minipage}

\subsection{Less Simple: many-choices non-determinism, used once}
\label{subsec:list-once}

If \code{choose} expects an \code{'a list} instead of just two
elements, the natural idea is to store, instead of a boolean, the
index of the list element to return. However, there is a difficulty:
at the time where \code{withNondeterminism} is called, it doesn't know
yet what the range of indices will be.

For a given program with a single \code{choose} operator, we'll refer to the list passed to \code{choose} as the \emph{choice list}. Because the code up until the \code{choose} call is deterministic, the choice list will be the same every time, so the program can simply remember its length in a piece of state.

Everything the implementation needs to know to achieve nondeterminism --- which item to choose on the next invocation, and when to stop --- can be captured in two pieces of state: the index to choose from, and the length of the choice list. We call this (index, length) pair a \emph{choice index}. On the first run, our implementation knows that it must pick the first choice (if there is one), but it doesn't know the length of the choice list. Hence, the global state is actually an option of a choice index, which starts at \code{NONE}.

\begin{nanoml}
type idx = int * int
val state : idx option ref = ref NONE
\end{nanoml}

\noindent So, for instance, in the  program \code{withNondeterminism (fn () => 2 * choose [1, 2, 3])}, 
the body will be run thrice, with states of \code{NONE}, \code{SOME (1, 3)}, and \code{SOME (2, 3)} respectively, instructing \code{choose} to select each item from the list.

We define some auxiliary functions on indices: to create the first index in a list, advance the index, and get the corresponding element from the choice list.

\begin{nanoml}
fun start_idx xs = (0, List.length xs)
fun next_idx (k, len) =
  if k + 1 = len then NONE
  else SOME (k + 1, len)
fun get xs (k, len) = List.nth (xs, k)
\end{nanoml}

\noindent We can now write the single-use \code{choose} function for arbitrary lists. An empty list aborts the computation with
an exception. Otherwise, it looks at the state. If it is already set
to some index, it returns the corresponding element. Otherwise it
initializes the state with the first index and returns this element.

\begin{minipage}{\textwidth}
\begin{nanoml}
exception Empty
\end{nanoml}
\begin{nanoml}
fun choose [] = raise Empty
  | choose xs = case !state of
                  NONE => let val i = start_idx xs in
                              state := SOME i;
                              get xs i
                          end
                | SOME i => get xs i
\end{nanoml}
\end{minipage}

\noindent The \code{withNondeterminism} function loops through every choice index, accumulating the results into a list. It returns an empty value if the choice list is empty.

\begin{nanoml}
fun withNondeterminism f =
  let val v = [f ()] handle Empty => [] in
    case !state of
      NONE => v
    | SOME i => case next_idx i of
                  NONE => v
                | SOME i' => (state := SOME i';
                              v @ withNondeterminism f)
  end
\end{nanoml}

\noindent Here's \code{withNondeterminism} in action:

\begin{nanoml}
- withNondeterminism (fn () => 2 * choose []);
val it = [] : int list
- withNondeterminism (fn () => 2 * choose [1, 2, 3]);
val it = [2,4,6] : int list
\end{nanoml}

\subsection{Several calls to \code{choose}}
\label{subsec:list-many}

In the previous implementation, a single choice index was sufficient to track all the choices made. To track several uses of \code{choose} in the body of \code{withNondeterminism}, we need a list of choice indices.

We can view a nondeterministic computation as a tree, where each call to
\code{choose} is a node, and each path corresponds to a sequence of choices. Our program must find every possible result of the nondeterministic computation --- it must find all leaves of the tree. The implementation is like a depth first search, except that it must replay the computation (i.e.: start from the root) once for each leaf in the tree. For example:

\begin{nanoml}
withNondeterminism (fn () =>
  if choose [true,false] then
    choose [5,6]
  else
    choose [7,8,9])
\end{nanoml}

\noindent There are five paths in the execution tree of this program. In the first run, the two calls to \code{choose} return $(\true,5)$. In the second, they return $(\true,6)$, followed by $(\false,7), (\false,8)$ and $(\false,9)$. Each path is identified by the sequence of choice indices of choices made during that execution, which we call a \emph{path index}. Our algorithm is built on a fundamental operation: to select a path index, and then execute the program down that path of the computation tree. To do this, it partitions the path index into two stacks during execution: the \code{past} stack contains the choices already made, while \code{future} contains the known choices to be made.

\begin{nanoml}
val past : idx list ref = ref []
val future : idx list ref = ref []

(* auxiliary stack functions *)
fun push stack x = (stack := x :: !stack)
fun pop stack =
  case !stack of
    [] => NONE
  | x :: xs => (stack := xs; SOME x)
\end{nanoml}

\noindent Figure \ref{fig:branching-nondet} depicts the execution tree for the last example, showing values of \code{past} and \code{future} at different points in the algorithm. To make them easier to update, path indices are stored as a stack, where the first element of the list represents the last choice to be made. For instance, the first path, in which the calls to \code{choose} return $true$ and $5$, has path index $[(0,2),(0,2)]$, and the next path, returning $true$ and $6$, has path index $[(1,2),(0,2)]$. These are shown in Figures \ref{fig:branching-nondet-a} and \ref{fig:branching-nondet-b}. The \code{next_path} function inputs a path index, and returns the index of the next path.

\begin{nanoml}
fun next_path [] = []
  | next_path (i :: is) =
    case next_idx i of
      SOME i' => i' :: is
    | NONE => next_path is
\end{nanoml}

\begin{figure}
\centering
\begin{subfigure}[t]{0.34\textwidth}
\centering
\includegraphics[scale=0.33]{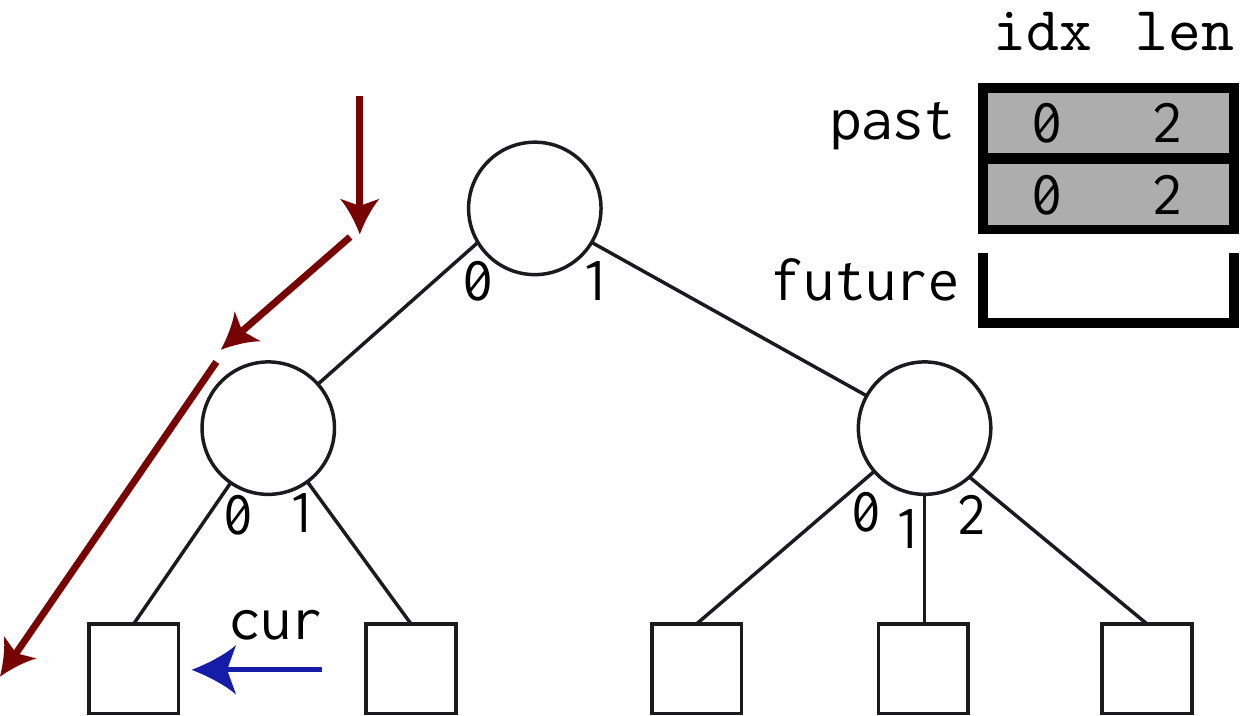}
\subcaption{}
\label{fig:branching-nondet-a}
\end{subfigure}
~
\begin{subfigure}[t]{0.31\textwidth}
\centering
\includegraphics[scale=0.33]{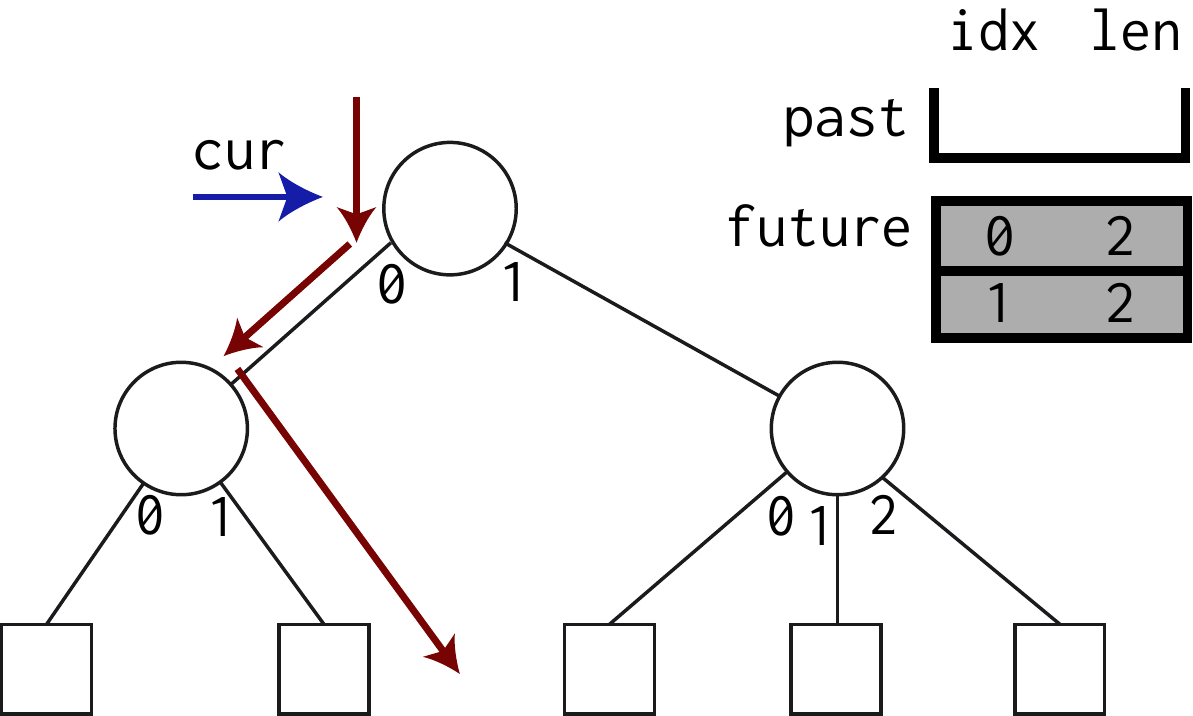}
\subcaption{}
\label{fig:branching-nondet-b}
\end{subfigure}
~
\begin{subfigure}[t]{0.31\textwidth}
\centering
\includegraphics[scale=0.33]{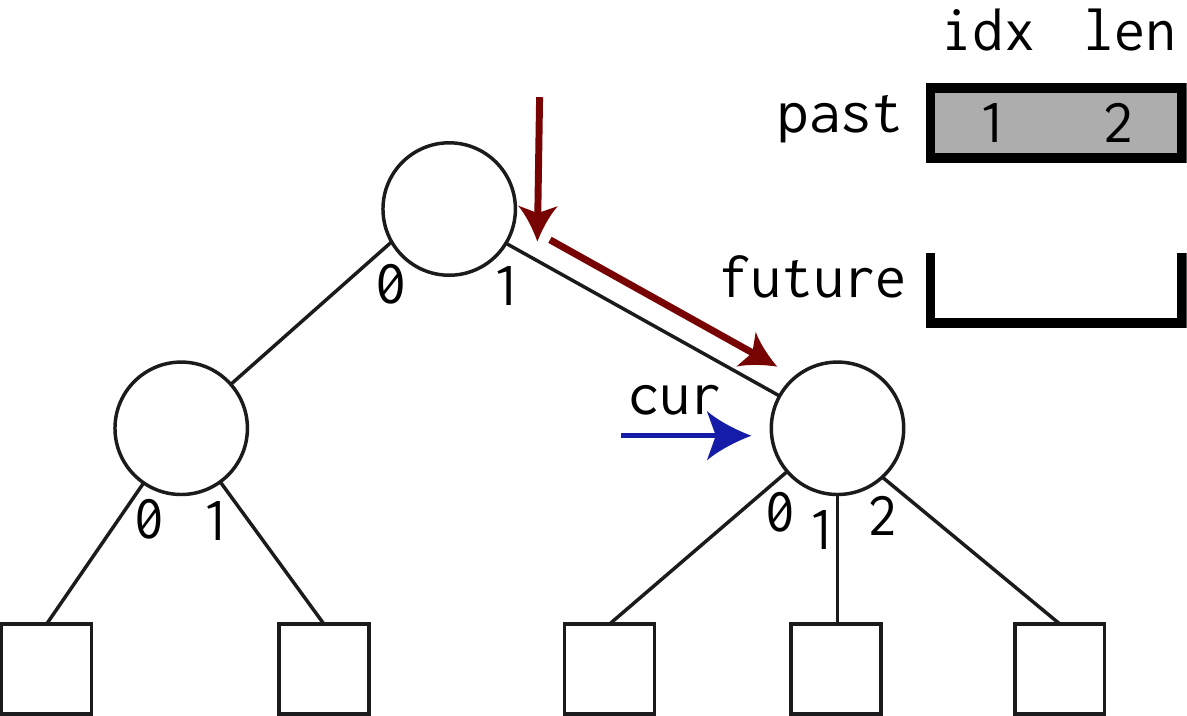}
\subcaption{}
\label{fig:branching-nondet-c}
\end{subfigure}

\caption{Several points in the execution of the replay-based nondeterminism algorithm.}
\label{fig:branching-nondet}
\end{figure}

\noindent When execution reaches a call to \code{choose}, it reads the choice to make from \code{future}, and pushes the result into the past. What if the future is unknown, as it will be in the first execution to reach a given call to \code{choose}? In this case, \code{choose} picks the first choice, and records it in \code{past}. Figure \ref{fig:branching-nondet-c} depicts this scenario for the first time execution enters the \code{else} branch of our example.

\begin{minipage}{\textwidth}
\begin{nanoml}
fun choose [] = raise Empty
  | choose xs = case pop future of
                  NONE => (* no future: start a new index; push it into the past *)
                          let val i = start_idx xs in
                            push past i;
                            get xs i
                          end
                | SOME i => (push past i;
                             get xs i)
\end{nanoml}
\end{minipage}

\noindent The execution of a path in the computation tree ends when a value is
returned. At this point, \code{future} is empty (all known \code{choose} calls have been executed), and \code{past} contains the complete path index. The \code{withNondeterminism} function is exactly the same as in
Section~\ref{subsec:list-once}, except that instead of updating the state to the next choice index, it computes the next path index, which becomes the \code{future} of the next run:

\begin{nanoml}
fun withNondeterminism f =
  let val v = [f ()] handle Empty => []
      val next_future = List.rev (next_path (!past))
  in
    past := [];
    future := next_future;
    if !future = [] then v
    else v @ withNondeterminism f
  end
\end{nanoml}

\noindent When \code{withNondeterminism} terminates, it must be the case that both \code{(!past) = []} and \code{(!future) = []}, which allows to run it again.
\begin{nanoml}
- withNondeterminism (fn () =>
       if choose [true, false] then choose [1, 2] else choose [3, 4]);
val it = [1,2,3,4] : int list
- withNondeterminism (fn () => 2 + choose [1, 2, 3] * choose [1, 10, 100]);
val it = [3,12,102,4,22,202,5,32,302] : int list
\end{nanoml}

\noindent There is still one thing missing: doing nested calls to \code{withNondeterminism} would overwrite \code{future} and \code{past}, so nested calls return incorrect results.

\begin{nanoml}
- withNondeterminism (fn () => if choose [true, false]
                               then withNondeterminism (fn () => choose [1, 2])
                               else []);
val it = [[1,1,2]] : int list list
\end{nanoml}

\subsection{Supporting nested calls}
\label{subsec:list-nested}

The final version of \code{withNondeterminism} supports nested calls
by saving the current values of \code{past} and \code{future} onto
a stack before each execution, and restoring them afterwards. No
change to \code{choose} is needed.

\begin{nanoml}
val nest : (idx list * idx list) list ref = ref []
exception Impossible

fun withNondeterminism f =
  (* before running, save current !past and !future value *)
  (push nest (!past, !future);
   past := [];
   future := [];
   let val result = loop f [] in
     (* then restore them after the run *)
     case pop nest of
       NONE => raise Impossible
     | SOME (p, f) => (past := p;
                       future := f;
                       result))
  end
       
and fun loop f acc =
  (* by the way, let's use a tail-recursive function *)
  let val acc = ([f ()] handle Empty => []) @ acc
      val next_future = List.rev (next_path (!past))
  in
    past := [];
    future := next_future;
    if !future = [] then acc
    else loop f acc
  end
\end{nanoml}

\section{Continuations in Disguise}
\label{sec:why-dc}

The previous section showed a trick for implementing direct-style nondeterminism in deterministic languages. Now, we delve to the deeper idea behind it, and surface the ability to generalize from nondeterminism to any monadic effect. We now examine how replay-based nondeterminism stealthily manipulates continuations. Consider evaluating this expression:

\begin{nanoml}
val e = withNondeterminism (fn () => choose [2,3,4] * choose [5, 6])
\end{nanoml}

\noindent Every subexpression of \code{e} has a continuation, and when it returns a value, it invokes that continuation. After the algorithm takes \code{e} down the first path and reaches the point \code{T = 2 * choose [5, 6]}, this second call to \code{choose} has continuation \code{C = (TEXT_ONLY_HOLE$\;$ => 2 * TEXT_ONLY_HOLE)}.

The \code{choose} function must invoke this continuation twice, with two different values. But
\code{C} is not a function that can be repeatedly invoked: it's a description of what the program does with a value after it's returned, and returning a value causes the program to keep executing, consuming the continuation. So \code{choose} invokes this continuation the first time normally, returning $5$. To copy this ephemeral continuation, it re-runs the computation until it's reached a point identical to \code{T}, evaluating that same call to \code{choose} with a second copy of \code{C} as its continuation --- and this time, it invokes the continuation with $6$.

So, the first action of the \code{choose} operator is capturing the continuation. And what happens next? The continuation is invoked once for each value, and the results are later appended together. We've already seen another operation that invokes a function once on each value of a list and appends the results: the \code{ListMonad.bind} operator. Figure \ref{fig:bind-cont} depicts how applying \code{ListMonad.bind} to the continuation produces direct-style nondeterminism.

So, replay-based nondeterminism is actually a fusion of two separate ideas:

\begin{enumerate}
\item Capturing the continuation using replay
\item Using the captured continuation with operators from the nondeterminism monad
\end{enumerate}

\noindent In Section \ref{sec:dc}, we extract the first half to create \emph{thermometer continuations}, our replay-based implementation of delimited control. The second half --- using continuations and monads to implement any effect in direct style --- is Filinski's construction, which we explain in Section \ref{sec:monadic-reflection}. These produce something more inefficient than the replay-based nondeterminism of Section \ref{sec:nondet}, but we'll show in Section \ref{sec:optimization} how to fuse them together into something equivalent.

\begin{figure}
\centering
\includegraphics[scale=0.3]{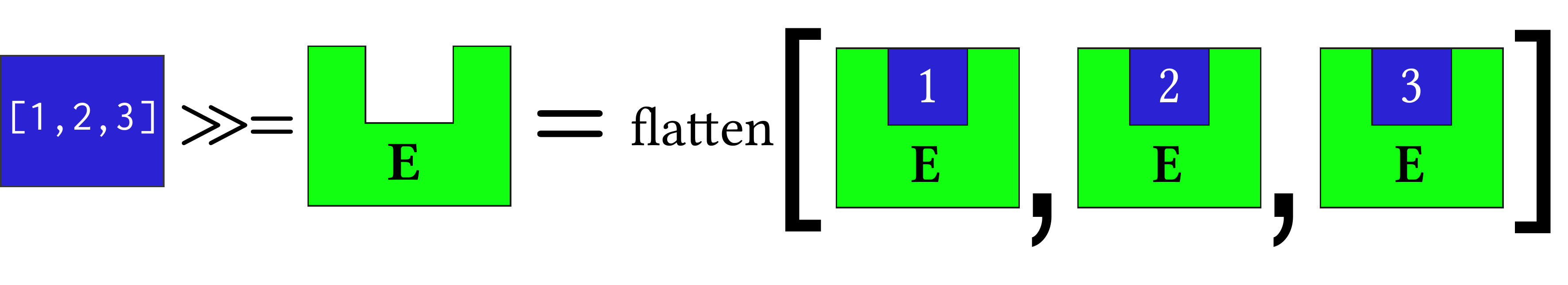}
\caption{To implement \code{choose}: first capture the continuation, and then use the list monad's \code{bind} operator to evaluate it multiple times.}
\label{fig:bind-cont}
\end{figure}

\section{Thermometer Continuations: Replay-Based Delimited Control}
\label{sec:dc}

In the previous section, we explained how the replay-based nondeterminism algorithm actually hides a mechanism for capturing continuations. Over the the remainder of this section, we extract out that mechanism, developing the more general idea of \emph{thermometer continuations}. But first, let us explain the variant of continuations that our mechanism uses: delimited continuations.

\subsection{What is delimited control?}
\label{sec:explains-dc}

When we speak of "the rest of the computation," a natural question is "until where?" For traditional continuations, the answer is: until the program halts. This crosses all abstraction boundaries, making these "undelimited continuations" difficult to work with. Another answer is: until some programmer-specified "delimiter" or "reset point." This is the answer of \emph{delimited continuations}, introduced by Felleisen,\scite{felleisen1988theory} which only represent a prefix of the remaining computation. Just as \code{call/cc} makes continuations first class, allowing a program to modify its continuation to implement many different global control-flow operators, the \code{shift} and \code{reset} constructs make delimited continuations first-class, and can be used to implement many local control-flow operators.

\newcommand{\reset}[0]{\square}

In the remainder of this section, we'll use the notation \code{$E$[x]} to denote plugging value $x$ into evaluation context \code{$E$}. We will also write
\code{$\reset$(t)} for the expression \code{reset (fun () => t)}. These notations makes it easy to give the semantics
for \code{shift} and \code{reset}.

If the body of a \code{reset} reaches a value \code{v}, then this value is returned to the outside context.
\begin{nanoml}
$E_1$[$\reset$(v)]   ==>   $E_1$[v]
\end{nanoml}

\noindent On the other hand, if the evaluation inside a \code{reset} reaches
a call to \code{shift}, then the argument of \code{shift} is invoked
with a delimited continuation \code{k} corresponding to the entire
evaluation context until the closest \code{reset} call:
\begin{nanoml}
$E_1$[$\reset$($E_2$[shift (fn k => t)])] ==>$\;$ $E_1$[$\reset$ (let fun k x = $\reset$($E_2$[x]) in t)]
\end{nanoml}

\noindent where $E_2$ does not contain contexts of the form
\code{$\reset$($E$)}. Figure \ref{fig:dc-switch} depicts this evaluation.

While \code{call/cc} captures an undelimited continuation that goes until the end of the program, \code{shift} only
captures the continuation up to the first enclosing \code{reset} -- it
is \emph{delimited}.

Notice that both the evaluation of \code{shift}'s argument \code{t}
and the body of its continuation \code{k} are wrapped in their own
\code{reset} calls; this affects the (quite delicate) semantics of
nested calls to \code{shift}. There exists other variants of control
operators that make different choices here.\scite{dyvbig2007monadic}

If the continuation \code{k} passed to \code{shift} is never called,
it means that the captured context $E_2$ is discarded and never used;
this use of \code{shift} corresponds to exceptions.
\begin{nanoml}
- reset (fn () => [1, 2] @ shift (fn k => [3, 4]))
==> reset (fn () => let fun k = x = [1, 2] @ x in reset (fn () => [3, 4]))
==> [3, 4]
\end{nanoml}

\noindent On the other hand, using a single continuation multiple times can be used to emulate \code{choose}-style non-determinism.
\begin{nanoml}
- reset (fn () => 3 * shift (fn k => [k 2, k 3, k 4]));
==> reset (fn () => let fun k x = 3 * x in reset (fn () => [k 2, k 3, k 4]))
==> [6, 9, 12]
\end{nanoml}

\begin{figure}
\centering
\includegraphics[scale=0.25]{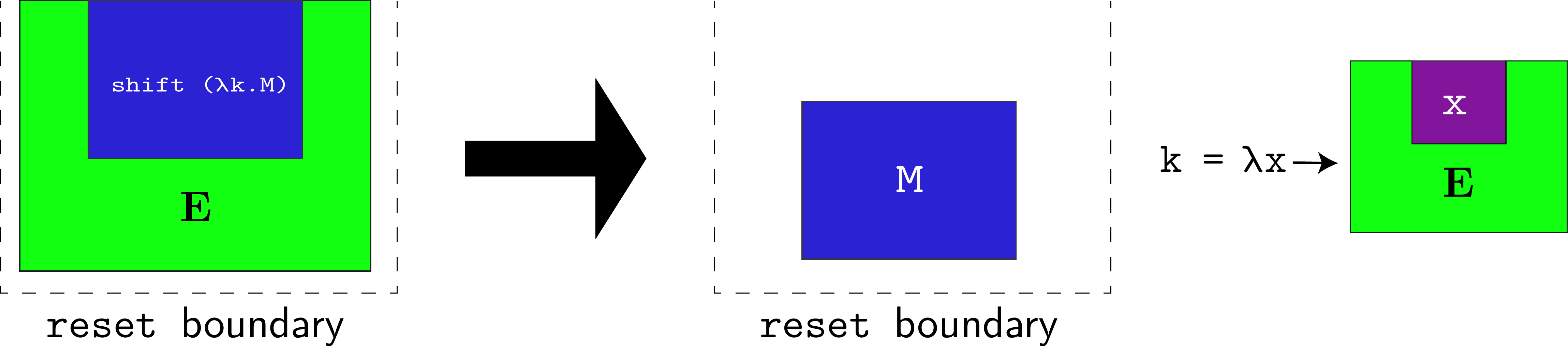}
\caption{Graphical depiction of the action of the \code{shift} operator.}
\label{fig:dc-switch}
\end{figure}

\noindent It is not easy to give precise polymorphic types to \code{shift} and
\code{reset}; in this work, we make the simplifying assumption that they
return a fixed ``answer'' type \code{ans}, as encoded in the following
SML signature.

\begin{nanoml}
signature CONTROL = sig
    type ans
    val reset : (unit -> ans) -> ans
    val shift : (('a -> ans) -> ans) -> 'a
end;
\end{nanoml}

\noindent See the tutorial of Asai and Kiselyov\scite{control-tutorial} for
a more complete introduction to delimited control operators.

\subsection{Baby Thermometer Continuations}
\label{sec:baby-thermo}

Programming with continuations requires being able to capture and copy an intermediate state of the program. We showed in Section \ref{sec:why-dc} that replay-based nondeterminism implicitly does this by replaying the whole computation, and hence needs no support from the runtime. We shall now see how to do this more explicitly.

This section presents a simplified version of thermometer continuations. This version assumes the \code{reset} block only contains one \code{shift}, which invokes its passed continuation $0$ or more times. It also restricts the body of \code{shift} and \code{reset} to only return integers. Still, this implementations makes it possible to handle several interesting examples. 

Whenever we evaluate a \code{reset}, we'll store its body as \code{cur_expr}:

\begin{nanoml}
exception Impossible
val cur_expr : (unit -> int) ref = ref (fn () => raise Impossible)

fun reset f = (cur_expr := f;
               (* rest of function given later *) )
\end{nanoml}

\noindent Now, let $F$ be some function containing a single \code{shift}, i.e.: \code{F = (fn () => E[shift (fn k => t))]}, where $E$ is a pure evaluation context (no additional effects). Then the main challenge in evaluating \code{reset F} is to capture the continuation of the \code{shift} as a function, i.e.: to obtain a \code{C = (fn x => E[x])}. 

Like the \code{future} stack in replay-nondeterminism that controls the action of \code{choose}, our trick is to use a piece of state to control \code{shift}.

\begin{nanoml}
val state : int option ref = ref NONE
\end{nanoml}

Suppose we implement \code{shift} so that \code{(state :=$\;$ SOME x; shift f)} evaluates to \code{x}, regardless of \code{f}. Now consider a function which sets the state, and then replays the \code{reset} body, e.g.: a function \code{k = (fn x => (state :=$\;$ SOME x; (!cur_expr) ()))}. Because \code{E} is pure, the effectful code \code{state :=$\;$  SOME x} commutes with everything in \code{E}. This means the following equivalences hold:

% HELP! I did not find a good way to put this side-by-side with explanations
%\begin{tabularx}{\textwidth}{lX}
\begin{nanoml}
k y
$\equiv$ (fn x => (state := SOME x; (!cur_expr) ())) y            $\hfill\text{\sffamily (by definition of \texttt{k})}$
$\equiv$ (state := SOME y; E[shift (fn k => t)])
$\equiv$ E[state := SOME y; shift (fn k => t)]                    $\hfill\text{\sffamily (\texttt{E} commutes with \texttt{state})}$
$\equiv$ E[y]                                                     $\hfill\text{\sffamily (\texttt{shift} property)}$
\end{nanoml}
%\end{tabularx}

\noindent which means that \code{k} is exactly the continuation we were looking for!

We are now ready to define \code{reset} and \code{shift}. \code{reset} sets the \code{cur_expr} to its body, and resets the \code{state} to \code{NONE}. It then runs its body. For the cases  like \code{reset (fn () => 1 + shift (fn () => 2))}, where the \code{shift} aborts the computation and returns directly to \code{reset}, the \code{shift} will raise an exception and \code{reset} will catch it.

\begin{nanoml}
exception Done of int
fun reset f = (cur_expr := f;
               state := NONE;
               (f () handle (Done x) => x))
\end{nanoml}

This gives us the definition of \code{shift f}: if \code{shift f} is being invoked from running the continuation, then \code{state} will not be \code{NONE}, and it should return the value in \code{state}. Else, it sets up a continuation as above, runs \code{f} with that continuation, and passes the result to the enclosing \code{reset} by raising an exception:

\begin{nanoml}
fun shift f = case !state of
                 SOME x => x
               | NONE => let val expr = !cur_expr
                             fun k x = (state := SOME x;
                                        expr ())
                             val result = f k
                          in
                              raise (Done result)
                          end
\end{nanoml}

\noindent We can now evaluate some examples with basic delimited continuations:

\begin{nanoml}
- reset (fn () => 2 * shift (fn k => 1 + k 5))
val it = 11 : int

- reset (fn () => 1 + shift (fn k => (k 1) * (k 2) * (k 3)))
val it = 24 : int
\end{nanoml}

\noindent The definitions given here lend themselves to a simple equational argument for correctness. Consider an arbitrary reset body with a single \code{shift} and no other effects. It reduces through these equivalences:

\begin{nanoml}

$E_1$[reset (fn () => $E_2$[shift (fn k => t)])]
==> $E_1$[$E_2$[raise (Done (let fun k x = (state := SOME x; $E_2$[shift (fn k => t)]) in t))]
         handle (Done x => x)]
$\equiv$ $E_1$[let fun k x = (state := SOME x; $E_2$[shift (fn k => t)]) in t]
$\equiv$ $E_1$[let fun k x = ($E_2$[state := SOME x; shift (fn k => t)]) in t]
$\equiv$ $E_1$[let fun k x = ($E_2$[x]) in t]
\end{nanoml}

\noindent Except for the nested \code{reset}'s, which do nothing in the case where there is only one \code{shift}, this is exactly the same as the semantics we gave in Section \ref{sec:explains-dc}.

\subsection{General Thermometer Continuations}
\label{sec:thermo-cont}

Despite its simplicity, the restricted delimited control implementation of Section \ref{sec:baby-thermo} already addressed the key challenge: capturing a continuation as a function. Just as we upgraded the single-choice nondeterminism into the full version, we'll now upgrade that implementation into something that can handle arbitrary code with delimited continuations. We explain this process in five steps:

\begin{enumerate}
\item Allowing answer types other than \code{int}
\item Allowing multiple sequential \code{shift}'s
\item Allowing different calls to \code{shift} to return different types
\item Allowing nested \code{shift}'s
\item Allowing nested \code{reset}'s
\end{enumerate}

\noindent We'll go through each in sequence.

\paragraph{Different answer types}

As mentioned in Section \ref{sec:explains-dc}, it's not easy to give precise polymorphic types to \code{shift} and \code{reset}. What can be done is to define them for any fixed answer type. An ML functor --- a "function" that constructs modules ---  is used to parametrize over different answer types.

\begin{minipage}{\textwidth}
\begin{nanoml}
functor Control (type ans) : CONTROL = struct
  type ans = ans
  val cur_expr : (unit -> ans) ref = ref (fn () => raise Impossible)
  (* ... all other definitions ... *)
end;
\end{nanoml}
\end{minipage}

\paragraph{Multiple sequential \code{shift}'s}

\begin{figure}[t]
\centering
\includegraphics[scale=0.5]{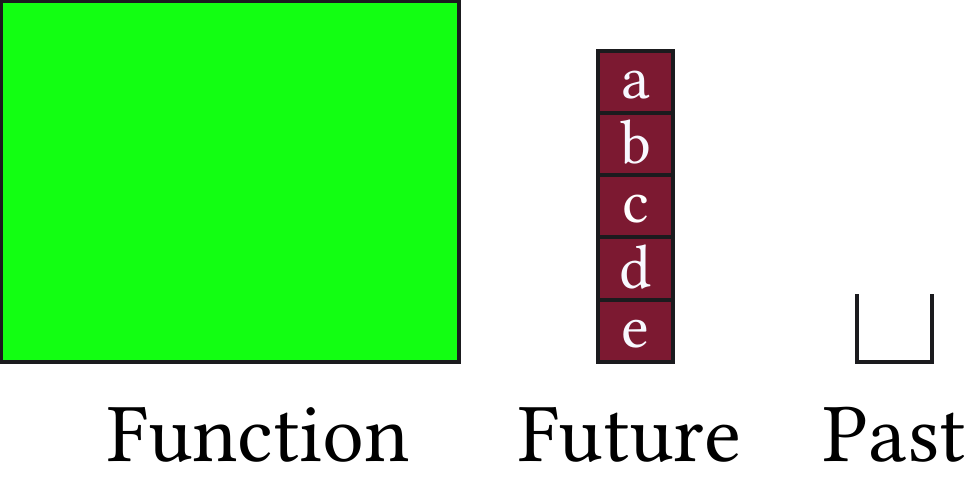}
\caption{A thermometer continuation before being invoked}
\label{fig:thermo-cont-init}
\end{figure}

\begin{figure}[t]
\centering
\includegraphics[scale=0.45]{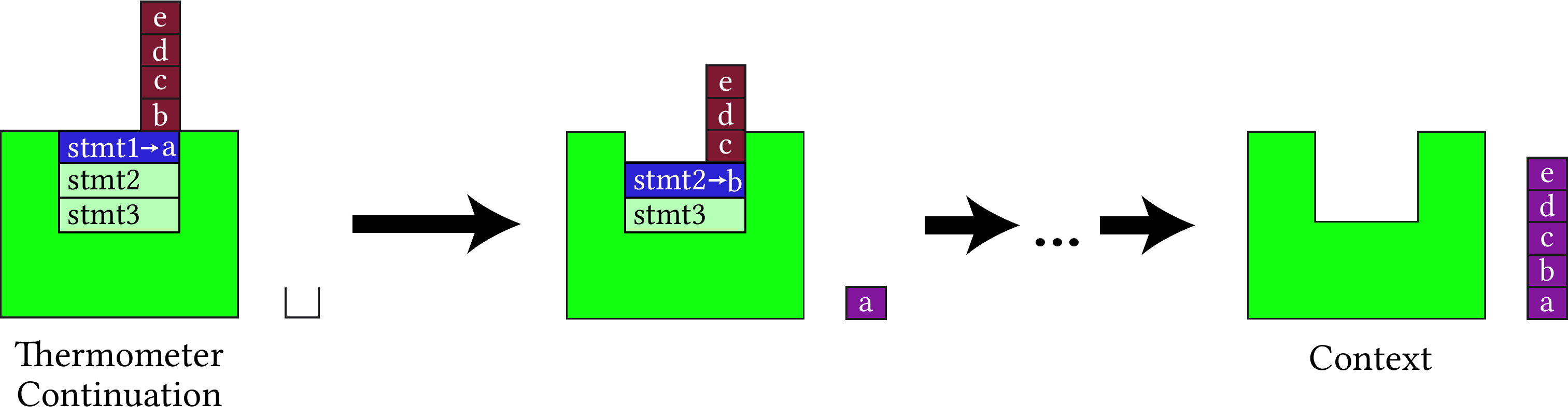}
\caption{Graphical depiction of running a thermometer continuation}
\label{fig:thermo-cont-exec}
\end{figure}

We will now generalize the construction of the previous section to multiple sequential \code{shift}'s. We do this similarly to what we did for replay nondeterminism, where we generalized the single choice index to a sequence of indices representing a path in the computation tree, and created the operation of sending execution along a certain path in the computation tree. This time, we generalize the simple state of the previous section which controls the return value of a \code{shift}, into a sequence of return values of sequential \code{shift}'s, making execution proceed in a certain fashion. Consider this example:

\begin{nanoml}
- reset (fn () => shift (fn k => 1 + k 2)
               * shift (fn k' => 1 + k' 3))
==> 8
\end{nanoml}

\noindent Suppose the state commands the first \code{shift} to return $x$ (call this Execution A). Then replaying the computation evaluates to \code{x * shift (fn k' => 1 + k' 3)} --- exactly the continuation of the first \code{shift}. Meanwhile, if the state commands the first \code{shift} to return $2$ (Execution B), and the second to return $x$, then replay will evaluate to \code{2 * x} --- this is the continuation of the second \code{shift}!

So, this state is a sequence of commands called \emph{frame}s, where each command instructs a \code{shift} to return a value (we'll soon add another kind of command). The sequence of values to be returned by future calls to \code{shift} is called the \emph{future}, so that Execution A was described by the future \code{[RETURN x]}, and Execution B by \code{[RETURN 2, RETURN x]}. This means that the continuation could be described by everything up to the \code{RETURN x}, namely the empty list \code{[]} and \code{[RETURN 2]}, respectively. This gives us our definition of a \emph{thermometer continuation}:

\begin{definition}
A {\bf thermometer continuation} for a continuation $C$ is a pair of a list and a computation, $(s, \texttt{block})$, so that for appropriate  \code{change_state} and \code{run}, the code \code{change_state x s; run block} evaluates to \code{$C$[x]}.
\end{definition}

\noindent Or, in code:

\begin{minipage}{\textwidth}
\begin{nanoml}
datatype frame = RETURN of int
type thermo_cont = (unit -> ans) * frame list
\end{nanoml}
\end{minipage}

\noindent The key operation of thermometer continuations is to run the computation using the associated future, making execution proceed to a certain point. Just as in replay nondeterminism, during execution, each item is moved from the future to the past as it's used:

\begin{nanoml}
val past : frame list ref = ref []
val future : frame list ref = ref []
\end{nanoml}

\noindent Figure \ref{fig:thermo-cont-init} depicts a thermometer continuation, while Figure \ref{fig:thermo-cont-exec} animates the execution of a thermometer continuation, showing how each value in the future matches with a \code{shift}, until the function being evaluated has the desired continuation. The left side of Figure \ref{fig:thermo-cont-exec} resembles a thermometer sticking out of the function, which inspired the name "thermometer continuations."

\paragraph{Different types for different \code{shift}'s}

A future stores the values to be returned by different \code{shift}'s. When these \code{shift}'s are used at different types, the future will hence be a list of heterogeneous types, which poses typing problems. Replay-nondeterminism dodged this problem by storing an index into the choice list instead of the value to return itself. That solution does not work here, where the value to be returned may be from a distant part of the program. In an expression \code{shift (fn k => t)}, the continuation \code{k} will only be used to pass values to future invocations of that same \code{shift}, which has the correct type. However, we cannot explain this to the ML type system.  Instead, we use a hack, a universal type \code{u} in which we assume that all types can be embedded and projected back:\footnote{See Section~\ref{subsec:universal-type} for a discussion of
our use of unsafe casts.}

\begin{nanoml}
signature UNIVERSAL = sig
    type u;
    val to_u : 'a -> u;
    val from_u : u -> 'a;
end;

structure Universal : UNIVERSAL = struct
  datatype u = U;
  val to_u = Unsafe.cast;
  val from_u = Unsafe.cast;
end;
\end{nanoml}

\noindent We use this universal type to define the real \code{frame} type below.

\paragraph{Nested \code{shift}'s}

With sequential \code{shift}'s, bringing execution to the desired point only required replacing \code{shift}'s with a return value. But with nested \code{shift}'s, the continuation to be captured may be inside another \code{shift}, and so replay will need to enter it. Consider this example:

\begin{nanoml}
- 1 + reset (fn () => 2 + shift (fn k => 3 * shift (fn l => l (k 10))))
val it = 37 : int
\end{nanoml}

\noindent The delimited continuation of the second \code{shift} is \code{C = (TEXT_ONLY_HOLE$\;$ => 3 * TEXT_ONLY_HOLE)}. In replay, execution must enter the body of the first \code{shift}, but replace the second \code{shift} with a value. So, in addition to \code{RETURN} frames, there must be another kind of \code{frame}, which instructs replay to enter the body of a \code{shift}.

\begin{nanoml}
datatype frame = RETURN of Universal.u | ENTER
\end{nanoml}

\noindent Then, the desired computation is expressed simply as \code{[ENTER]}, and can be evaluated with value \code{x} using the future \code{[ENTER, RETURN x]}. The values \code{past} and \code{future} are stacks of this \code{frame} type.

\begin{nanoml}
val past : frame list ref = ref []
val future : frame list ref = ref []
\end{nanoml}

\paragraph{Nested \code{reset}'s}

This is handled completely analogously to Section \ref{subsec:list-nested}: each call to \code{reset} will push the previous state to a nesting stack, and restore it afterwards.

\begin{nanoml}
type reset_state = (unit -> ans) * state list * state list
val nest : reset_state list ref = ref []
\end{nanoml}

\noindent We are now ready to begin giving the final implementation of thermometer continuations. As before, \code{shift} will pass its value to \code{reset} by raising an exception.

\begin{nanoml}
exception Done of ans
\end{nanoml}

\noindent The key operation is to play a computation according to a known future. We implement this in the \code{run_with_future} function, which both does that, and establishes a new \code{reset boundary}. When doing so, it must save and restore the current state. This is similar to the \code{withNondeterminism} version with nesting, with the added management of \code{cur_expr}.

\begin{nanoml}
fun run_with_future f f_future =
  (push nest (!cur_expr, !past, !future);
   past := [];
   future := f_future;
   cur_expr := f;
   let val result = (f () handle (Done x) => x) in
     case pop nest of
       NONE => raise Impossible
     | SOME (prev_expr, prev_past, prev_future) =>
       (cur_expr := prev_expr;
        past := prev_past;
        future := prev_future;
        result)
   end)
\end{nanoml}

\noindent The \code{reset} we expose to the user is just a specialized instance
of \code{run_with_future}, running with an empty future:

\begin{nanoml}
fun reset f = run_with_future f []
\end{nanoml}

\noindent Finally, \code{shift} is similar to \code{choose}. The easy case is
when the next item in the future is a return frame; it just has to return the commanded value.

\begin{nanoml}
fun shift f = case pop future of
    SOME (RETURN v) =>
      (push past (RETURN v);
       Universal.from_u v)
  | (NONE | SOME ENTER) =>
    let val new_future = List.rev (!past)
        val our_expr = !cur_expr
        fun k v = run_with_future our_expr (new_future @ [RETURN (Universal.to_u v)])
        val () = push past ENTER
        val result = f k
    in
        raise (Done result)
    end
\end{nanoml}

\noindent The two other cases are more delicate. The \code{ENTER} case
corresponds to the case where execution must enter the
\code{shift} body --- the function \code{f}. The case \code{NONE} is
when the future is unknown, but, in that case, it should do exactly
the same: enter the \code{shift} body.

It records this decision by pushing an \code{ENTER} frame into the past, and prepares
a function \code{k} that invokes the thermometer continuation defined by \code{(!cur_expr, List.rev (!past))}. It invokes the thermometer continuation by appending a new value to the future, and then calling \code{run_with_future}. The body of the \code{shift} executes with this \code{k}; if it terminates normally, it raises an exception to pass the result to the enclosing \code{reset}.

\begin{nanoml}
- 1 + reset (fn () => 2 + shift (fn k => 3 * shift (fn l => l (k 10))))
==> 37
\end{nanoml}

\subsection{What language features does this need?}

Thermometer continuations clearly require exceptions and state, but the implementation of this section uses a couple other language features as well. Here, we discuss why we chose to use them, and ideas for alternatives.

\paragraph{Why we need \code{$\;$Unsafe.cast} for universal types}
\label{subsec:universal-type}

Our implementation of thermometer continuations uses \code{Unsafe.cast}, as in the
original presentation of Filinski,\scite{Filinski94} but \code{shift}
and \code{reset} are carefully designed so that these casts never
fail. Similarly, we believe that our use of \code{Unsafe.cast} is
actually type-safe: we only inject values into the universal type at
the call site of an effectful function, and only project values back in
the same call site they are coming from.

In his follow-up paper,\scite{filinski1999} Filinski was able to
upgrade these definitions to a type-safe implementation of the
universal type.\footnote{See several implementations at
  http://mlton.org/UniversalType} Unfortunately, we cannot just
replace our implementation with a safe one.

Safe implementations of universal types cannot provide a pair of
uniform polymorphic embedding and projection functions \code{('a -> u)
  * (u -> 'a)} -- this is inherently unsafe. Instead, they provide
a function that, on each call, generates a new embedding/projection
pair for a given type: \code{unit -> (('a -> u) * (u -> 'a))}. They
all use a form of dynamic generativity, such as the allocation of new
memory, or the creation of a new exception constructor.

These implementations do not work properly with our replay-based
technique: each time a computation is replayed, a fresh
embedding/projection pair is generated, and our implementation then
tries to use the projection of a new pair on a value (of the
same type) embedded by an old pair, which fails at runtime.

It may be possible to integrate the generation of universal embeddings
within our replay machinery, to obtain safe implementations, but we
left this delicate enhancement to future work.

\paragraph{Garbage Collection and Closures}

This implementation uses closures heavily. This is a small barrier to someone implementing thermometer continuations in C, as they must encode all closures as function pointer/environment pairs. However, there's a bigger problem: how to properly do memory management for these environments? A tempting idea is to release all environments after the last \code{reset} terminates, but this doesn't work: the captured continuations may have unbounded lifetime, as exemplified by the state monad implementation in Section \ref{sec:state-monad}.

The options to solve this problem are the same as when designing other libraries with manual memory management: either make the library specialized enough that it can guess memory needs, or provide extra APIs so the programmer can indicate them. For example, one possibility for the latter is to allocate all closures for a given \code{reset} block in the same arena. At some point, the programmer knows that the code will no longer need any of the continuations captured within the \code{reset}. The programmer invokes an API to indicate this, and all memory in that arena is released.

\section{Arbitrary Monads}
\label{sec:monadic-reflection}

In 1994, Filinski showed how to use delimited continuations to express any monadic effect in direct style.\footnote{\citet*{Filinski94}; see the blog post of \citet*{piponi_mother} for an introduction.} In this section, we hope to convey an intuition for Filinski's construction, and also discuss what it looks like when combined with thermometer continuations. The code in this section comes almost verbatim from Filinski. This section is helpful for understanding the optimizations of Section \ref{sec:optimization}, in which we explain how to fuse thermometer continuations with the code in this section.

\subsection{Monadic Reflection}

In SML and Java, there are two ways to program with mutable state. The first is to use the language's built-in variables and assignment. The second is to use the monadic encoding, programming similar to how a pure language like Haskell handles mutable state. A stateful computation is a monadic value, a pure value of type $s \rightarrow (a, s)$.

These two approaches are interconvertible. A program can take a value of type $s \rightarrow (a, s)$ and run it, yielding a stateful computation of return type $a$. This operation is called \code{reflect}. Conversely, it can take a stateful computation of type $a$, and \code{reify} it into a pure value of type $s \rightarrow (a,s)$.  Together, the \code{reflect} and \code{reify} operations give a correspondence between monadic values and effectful computations. This correspondence is termed \emph{monadic reflection}. Fillinski showed how, using delimited control, it is possible to encode these operations as program terms.

The \code{reflect} and \code{reify} operations generalize to arbitrary monads. Consider nondeterminism, where a nondeterministic computation is either an effectful computation of type \code{'a}, or a monadic value of type \code{'a list}. Then the \code{reflect} operator would take the input $[1,2,3]$ and nondeterministically return $1$, $2$, or $3$ --- this is the \code{choose} operator from Section \ref{sec:nondet}). So \code{reify} would take a computation that nondeterministically returns $1$, $2$, or $3$, and return the pure value $[1,2,3]$ --- this is \code{withNondeterminism}.

So, for languages which natively support an effect, \code{reflect} and \code{reify} convert between  effects implemented by the semantics of the language, and effects implemented within the language. Curry is a language with built-in nondeterminism, and it has these operators, calling them \code{anyOf} and \code{getAllValues}. SML does not have built-in nondeterminism, but, for our previous example, one can think of the code within a \code{withNondeterminism} block as running in a language extended with nondeterminism. So, one can think of the construction in the next section as being able to extend a language with any monadic effect.

In SML, monadic reflection is given by the following signature:

\begin{minipage}{\textwidth}
\begin{nanoml}
signature RMONAD = sig
    structure M : MONAD
    val reflect : 'a M.m -> 'a
    val reify : (unit -> 'a) -> 'a M.m
end;
\end{nanoml}
\end{minipage}

\subsection{Monadic Reflection through Delimited Control}
\label{sec:implements-reflection}

Filinski's insight was that the monadic style is similar to an older concept called continuation-passing style. We can see this by revisiting an example from Section \ref{sec:contains-list-monad}. Consider this expression:

\begin{nanoml}
withNondeterminism (fun () => (choose [2,3,4]) * (choose [5,6]))
\end{nanoml}

\noindent It is transformed into the monadic style as follows:

\begin{nanoml}
bind [2,3,4] (fn x =>
bind [5,6]   (fn y =>
  return (x * y)))
\end{nanoml}

\noindent The first call to \code{choose} has continuation \code{fn TEXT_ONLY_HOLE => TEXT_ONLY_HOLE * (choose [5,6])}. If \code{x} is the value returned by the first call to \code{choose}, the second has continuation \code{fn TEXT_ONLY_HOLE => x * TEXT_ONLY_HOLE}. These continuations correspond exactly to the functions bound in the monadic style. The monadic \code{bind} is the "glue" between a value and its continuation. Nondeterministically choosing from \code{[2,3,4]} wants to return thrice, which is the same as invoking the continuation thrice, which is the same as binding to the continuation.

So, converting a program to monadic style is quite similar to converting a program to this "continuation-passing style." Does this mean a language that has continuations can program with monads in direct style? Filinski answers yes.

The definition of monadic reflection in terms of delimited control is short. The overall setup is as follows:

\begin{minipage}{\textwidth}
\begin{nanoml}
functor Represent (M : MONAD) : RMONAD = struct
  structure C = Control(type ans = Universal.u M.m)
  structure M = M

  fun reflect m = ...
  fun reify t = ...
end;
\end{nanoml}
\end{minipage}

\noindent Figure \ref{fig:bind-cont} showed how nondeterminism can be implemented by binding a value to the (delimited) continuation. The definition of \code{reflect} is a straightforward generalization of this.

\begin{nanoml}
  fun reflect m = C.shift (fn k => M.bind m k)
\end{nanoml}

\noindent If \code{reflect} uses \code{shift}, then \code{reify} uses \code{reset} to delimit the effects implemented through shift. This implementation requires use of type casts, because \code{reset} is monomorphized to return a value of type \code{Universal.u m}. Without these type casts, \code{reify} would read

\begin{nanoml}
  fun reify t = C.reset (fn () => M.return (t ()))
\end{nanoml}

\noindent Because of the casts, the actual definition of \code{reify} is slightly more complicated:

\begin{nanoml}
  fun reify t = M.bind (C.reset (fn () => M.return (Universal.to_u (t ()))))
                       (M.return o Universal.from_u)
\end{nanoml}

\subsection{Example: Nondeterminism}

\label{sec:delim-nondet}

Using this general construction, we immediately obtain an implementation of nondeterminism equivalent to the one in Section \ref{sec:nondet} by using \code{ListMonad} (defined in Section \ref{sec:contains-list-monad}).

\begin{nanoml}
structure N = Represent(ListMonad)

fun choose xs = N.reflect xs
fun fail () = choose []

- N.reify (fn () => let val x = choose [2,3,4] * choose [5,7] in
                    if x >= 20 then x
                    else fail () end)

(* val it = [21,20,28] : int list *)
\end{nanoml}

\noindent It's worth thinking about how this generic implementation executes on the example, and contrasting it with the direct implementation of Section \ref{sec:nondet}. The direct implementation executes the function body $6$ times, once for each path of the computation. The generic one executes the function body $10$ times (once with a future stack of length $0$, $3$ times with length $1$, and $6$ times with length $2$). In the direct implementation, \code{choose} will return a value if it can. In the generic one, \code{choose} never returns. Instead, it invokes the thermometer continuation, causes the desired value to be returned at the equivalent point in the computation, and then raises an exception containing the final result. So, $4$ of those times, it could just return a value rather than replaying the computation. This is the idea of one of the optimizations we discuss in Section \ref{sec:optimization}. This, plus one other optimization, let us derive the direct implementation from the generic one.

\subsection{Example: State monad}
\label{sec:state-monad}

State implemented through delimited control works differently from SML's native support for state.

\begin{minipage}{\textwidth}
\begin{nanoml}
functor StateMonad (type state) : MONAD = struct
  type 'a m = state -> 'a * state
  fun return x = fn s => (x, s)
  fun bind m f = fn s => let val (x, s') = m s
                         in f x s' end
end;

structure S = Represent(StateMonad(type state = int))

fun tick () = S.reflect (fn s => ((), s+1))
fun get ()  = S.reflect (fn s => (s, s))
fun put n   = S.reflect (fn _ => ((), n))

- #1 (S.reify (fn () => (put 5; tick ();
                         2 * get ()))
               0)

(* val it = 12 : int *)
\end{nanoml}
\end{minipage}

\noindent Let's take a look at how this works, starting with the example \code{reify (fn () => 3 * get ())}.

\begin{nanoml}
- (reify (fn () => 3 * get ())) 2
===> (reify (fn () => 3 * (reflect (fn s => (s, s))))) 2
===> (reset (fn () => return (3 * (shift (fn k => bind (fn s => (s, s)) k))))) 2
===> (fn k => bind (fn s => (s,s)) k end)(fn x => return (3*x)) 2
===> (let val k = (fn x => fn s => (3*x, s)) in (fn s => k s s)) 2
===> (fn s => (3*s, s)) 2
===> (6, 2)
\end{nanoml}

\noindent The \code{get} in \code{reify (fn () => 3 * get ())} suspends the current computation, causing the \code{reify} to return a function which awaits the initial state. Once invoked with an initial state, it resumes the computation (multiplying by $3$).

What does \code{reify (fn () => (tick (); get ()))} do?
  The call to \code{tick ()} again suspends the computation and awaits
  the state, as we can see from its expansion \code{$\:$shift (fn k => fn s => k () (s+1))}. Once it receives \code{s}, it resumes it, returning \code{()} from \code{tick}. The call to \code{get} suspends the computation again, returning a function that awaits a new state; \code{tick} supplies \code{s+1}.

Think for a second about how this works when \code{shift} and \code{reset} are implemented as thermometer continuations. The \code{get}, \code{put}, and \code{tick} operators do not communicate by mutating state. They communicate by suspending the computation, i.e.: by raising exceptions containing functions. So, although the implementation of state in terms of thermometer continuations uses SML's native support for state under the hood, it only does so tangentially, to capture the continuation.

\section{Optimizations}
\label{sec:optimization}

%DONE: Add a subsection clarifying how much these optimizations buy you. Give table showing optimization results. Implement n-queens in SML/Ocaml with all three variations of nondet, and also done the normal way, as well as in Curry .  Maybe also Scala with delimited
%
%
%* Summing a list of ints (in list monad)
%* Summing a list of ints, throw away invalid, each one parsed in list monad
%* Something in the StateT Maybe / MaybeT State

Section \ref{sec:delim-nondet} compared the two implementations of nondeterminism, and found that the generic one using thermometer continuations replayed the computation gratuitously. Thermometer continuations also replay the program in nested fashion, consuming stack space. In this section, we sketch a few optimizations that make monadic reflection via thermometer continuations less impractical, and illustrate the connections between the two implementations of nondeterminism.

All optimizations in this section save memoization require fusing thermometer continuations with monadic reflection. So, instead of using the definition of \code{reflect} in terms of \code{shift} given in Section \ref{sec:implements-reflection}, the require defining a single combined \code{reflect} operator.

\subsection{CPS-bind: Invoking the Continuation at the Top of the Stack}
\label{sec:cont-top-stack}

The basic implementation of thermometer continuations wastes stack space. When there are multiple \code{shift}'s, execution will reach the first \code{shift}, then replay the function and reach the second \code{shift}, then replay again, etc, all the while letting the stack deepen. And yet, at the end, it will raise an exception that discards most of the computation on the stack. So, the implementation could save a lot of stack space by raising an exception before replaying the computation.

%TODO: The last sentence above is redundant with the next

So, when a program invokes a thermometer continuation, it will need to raise an exception to transfer control to the enclosing \code{reset}, and thereby signal \code{reset} to replay the computation. While the existing \code{Done} exception signals that a computation is complete, it can do this with a second kind of exception, which we call \code{Invoke}.

However, the \code{shift} and \code{reset} functions do not invoke a thermometer continuation: the body of the \code{shift} does. In the case of monadic reflection, this is the monad's \code{bind} operator. Raising an \code{Invoke} exception will discard the remainder of \code{bind}, so it must somehow also capture the continuation of \code{bind}.  We can do this by writing \code{bind} itself in the continuation-passing style, i.e.: with the following signature:

\begin{nanoml}
type ('b, 'c) cont = ('b -> 'c) -> 'c
val bind : 'a m -> ('a -> ('b m, 'c) cont) -> ('b m, 'c) cont
\end{nanoml}

\noindent The supplementary material contains code with this optimization, and uses it to implement nondeterminism in a way that executes more similarly to the direct implementation. We give here some key points. Here's what the CPS'd \code{bind} operator for the list monad looks like:

\begin{minipage}{\textwidth}
\begin{nanoml}
fun bind []      f d = d []
  | bind (x::xs) f d = f x (fn a => bind xs f (fn b => d (a @ b)))
\end{nanoml}
\end{minipage}

\noindent When used by \code{reflect}, \code{f} becomes a function that raises the \code{Invoke} exception, transferring control to the enclosing \code{reset}, which then replays the entire computation, but at the top level. The continuations of the \code{bind} get nested in a manner which is harder to describe, but ultimately get evaluated at the very end, also at the top level. So the list appends in \code{d (a @ b)} actually run at the top level of the \code{reset}, similar to how, in direct nondeterminism, it is the outer call to \code{withNondeterminism} that aggregates the results of each path.

While this CPS-monad optimization as described here can be used to implement many monadic effects, it cannot be used to implement all of them, nor general delimited continuations. Consider the state monad from Section \ref{sec:state-monad}: \code{bind} actually returns a function which escapes the outer \code{reify}. Then, when the program invokes that function and it tries to invoke its captured thermometer continuation, it will try to raise an \code{Invoke} exception to transfer control to its enclosing \code{reify}, but there is none. This CPS-monad optimization as described does not work if the captured continuation can escape the enclosing \code{reset}. With more work, it could use mutable state to track whether it is still inside a \code{reset} block, and then choose to raise an \code{Invoke} exception or invoke the thermometer continuation directly.

\subsection{Direct Returns}

In our implementation, a \code{reset} body \code{E[reflect (return 1)]} expands into

\begin{nanoml}
E[raise (Done (E[1] handle (Done x => x)))]
\end{nanoml}

\noindent So, the entire computation up until that \code{reflect} runs twice. Instead of replaying the entire computation, that \code{reflect} could just return a value. \code{E[reflect (return 1)]} could expand into \code{E[1]}.

The expression \code{reflect (return 1)} expands into \code{shift (fn k => bind (return 1) k)}. By the monad laws, this is equivalent to \code{shift (fn k => k 1)}. Tail-calling a continuation is the same as returning a value, so this is equivalent to \code{1}. So, it's the tail-call that allows this instance of \code{reflect} to return a value instead of replaying the computation.

Implementing the direct-return optimization is an additional small tweak to the CPS-\code{bind} optimization. We duplicate the second argument of bind, resulting in two function arguments that expect inputs at type \code{'a}.

\begin{nanoml}
val bind : 'a m -> ('a -> ('b m) cont) -> ('a -> ('b m) cont) -> ('b m -> 'c) -> 'c
\end{nanoml}

\noindent To the implementor of \code{bind}, we ask that the first function be invoked on the first value of type \code{'a} (if any) extracted out of the \code{'a m} output, and the second be used on all later values.

In our implementation of \code{reflect}, we use this more flexible \code{bind} type for optimization. The second function argument we pass raises an \code{Invoke} exception, as described in Section \ref{sec:cont-top-stack}. The first function returns a value directly, after updating the internal state of the thermometer. So, the first time \code{bind} invokes the continuation, it may do so by directly returning a value. Thereafter, it must invoke the continuation by explicitly raising an \code{Invoke} exception.

The \code{bind} operator for the list monad never performs a tail-call (it must always wrap the result in a list), but, after converting it to CPS, it always performs a tail-call. So this direct-return optimization combines well with the previous CPS-monad optimization. Indeed, applying them both transforms the generic nondeterminism of Section \ref{sec:delim-nondet} into the direct nondeterminism of Section \ref{sec:nondet}. In Section \ref{sec:benchmarks}, we show benchmarks showing that this actually gives an  implementation of nondeterminism as fast as the code in Section \ref{sec:nondet}.

In the supplementary material, we demonstrate this optimization, providing optimized implementations of nondeterminism (list monad) and failure (maybe monad).

\subsection{Memoization}

While the frequent replays of thermometer continuations can interfere with any other effects in a computation, it cannot interfere with observationally-pure memoization. Memoizing nested calls to \code{reset} can save a lot of computation, and any expensive function can memoize itself without integrating into the implementation of thermometer continuations.

\section{Performance numbers}
\label{sec:benchmarks}

The idea of invoking continuations by replaying the past at first
seems woefully impractical. In fact, it performs surprisingly well in
a diverse set of benchmarks. This shows that one may seriously
consider using this approach to conveniently write
non-performance-critical applications in a runtime (OCaml, Java,
Javascript...)  that does not provide control operators.

Working code for all benchmarks is available from
\url{https://github.com/jkoppel/thermometer-continuations}.

\subsection{The worst case}

The overhead of replaying pure computations depends on the ratio, in
an effectful program, of computation time spent in pure and impure
computations, and in the branching structure of the computation tree.

The worst case is when a program runs a costly pure computation,
followed by large branching of an effect -- causing many replays. The
overhead may be arbitrarily large; for example, replay-based or
thermometer-based implementations cause a 10x slowdown on the
following program:
\begin{nanoml}
withNondeterminism (fun () =>
  let val v = long_pure_computation () in
      val i = choose [0,1,2,3,4,5,6,7,8,9] in
    (i, v)
  end
)
\end{nanoml}

\subsection{Search-heavy programs: N-queens}

\begin{table}
\begin{center}

\caption{Benchmark \textsc{NQUEENS}}
\begin{tabular}{r|rrrrr}
\hline
OCaml times & \textbf{10} & \textbf{11} & \textbf{12} & \textbf{13} \\
\hline
\textbf{Indirect }
& 0.007s
& 0.037s
& 0.213s
& 1.308s
\\
\textbf{Replay }
& 0.017s
& 0.101s
& 0.597s
& 3.768s
\\
\textbf{Therm. }
& 0.041s
& 0.221s
& 1.347s
& 8.621s
\\
\textbf{Therm. Opt. }
& 0.019s
& 0.111s
& 0.65s
& 3.924s
\\
\textbf{Filinski (Delimcc)}
& 0.035s
& 0.185s
& 1.236s
& 14.412s
\\
\textbf{Eff. Handlers (Multicore OCaml)}
& 0.019s
& 0.092s
& 0.509s
& 2.81s
\\

\hline
SML times & \textbf{10} & \textbf{11} & \textbf{12} & \textbf{13} \\
\hline
\textbf{Indirect }
& 0.011s
& 0.064s
& 0.36s
& 2.166s
\\
\textbf{Replay }
& 0.029s
& 0.164s
& 1.051s
& 6.562s
\\
\textbf{Therm. }
& 0.06s
& 0.344s
& 2.152s
& 14.305s
\\
\textbf{Therm. Opt. }
& 0.03s
& 0.173s
& 1.151s
& 6.793s
\\
\textbf{Filinski (Call/cc)}
& 0.015s
& 0.079s
& 0.493s
& 2.941s
\\

\hline
MLton times & \textbf{10} & \textbf{11} & \textbf{12} & \textbf{13} \\
\hline
\textbf{Indirect}
& 0.007s
& 0.031s
& 0.187s
& 1.129s
\\
\textbf{Replay}
& 0.016s
& 0.099s
& 0.637s
& 4.295s
\\
\textbf{Finlinski (Call/cc)}
& 0.078s
& 0.399s
& 2.138s
& 12.62s
\\

\hline
Prolog times & \textbf{10} & \textbf{11} & \textbf{12} & \textbf{13} \\
\hline
\textbf{Prolog search (GNU Prolog)}
& 0.165s
& 0.614s
& 3.307s
& 20.401s
\\

\hline
\end{tabular}%
\hspace{1.23cm} % account for the missing column

\begin{tikzpicture}
\begin{semilogyaxis}[
    title={OCaml benchmarks},
    width={7cm},
    legend pos=outer north east,
    ymajorgrids=true,
    grid style=dashed,
    ytick={0.01,0.1,1,10,360},
    yticklabels={10ms,100ms,1s,10s,5m},
]

\addplot[color=black,mark=-]
coordinates {
(10, 0.007)
(11, 0.037)
(12, 0.213)
(13, 1.308)
};
\addlegendentry{Indirect }

\addplot[color=orange,mark=triangle]
coordinates {
(10, 0.017)
(11, 0.101)
(12, 0.597)
(13, 3.768)
};
\addlegendentry{Replay }

\addplot[color=red,mark=x]
coordinates {
(10, 0.041)
(11, 0.221)
(12, 1.347)
(13, 8.621)
};
\addlegendentry{Therm. }

\addplot[color=purple,mark=|]
coordinates {
(10, 0.019)
(11, 0.111)
(12, 0.65)
(13, 3.924)
};
\addlegendentry{Therm. Opt. }

\addplot[color=blue,mark=star]
coordinates {
(10, 0.035)
(11, 0.185)
(12, 1.236)
(13, 14.412)
};
\addlegendentry{Filinski (Delimcc)}

\addplot[color=brown,mark=asterisk]
coordinates {
(10, 0.019)
(11, 0.092)
(12, 0.509)
(13, 2.81)
};
\addlegendentry{Eff. Handlers (Multicore OCaml)}

\end{semilogyaxis}
\end{tikzpicture}

\begin{tikzpicture}
\begin{semilogyaxis}[
    title={SML/NJ benchmarks},
    width={7cm},
    legend pos=outer north east,
    ymajorgrids=true,
    grid style=dashed,
    ytick={0.01,0.1,1,10,360},
    yticklabels={10ms,100ms,1s,10s,5m},
]

\addplot[color=black,mark=-]
coordinates {
(10, 0.011)
(11, 0.064)
(12, 0.36)
(13, 2.166)
};
\addlegendentry{Indirect }

\addplot[color=orange,mark=triangle]
coordinates {
(10, 0.029)
(11, 0.164)
(12, 1.051)
(13, 6.562)
};
\addlegendentry{Replay }

\addplot[color=red,mark=x]
coordinates {
(10, 0.06)
(11, 0.344)
(12, 2.152)
(13, 14.305)
};
\addlegendentry{Therm. }

\addplot[color=purple,mark=|]
coordinates {
(10, 0.03)
(11, 0.173)
(12, 1.151)
(13, 6.793)
};
\addlegendentry{Therm. Opt. }

\addplot[color=blue,mark=star]
coordinates {
(10, 0.015)
(11, 0.079)
(12, 0.493)
(13, 2.941)
};
\addlegendentry{Filinski (Call/cc)}

\end{semilogyaxis}
\end{tikzpicture}

\label{table:nqueens}
\end{center}
\end{table}

The benchmark \textsc{NQUEENS} is the problem of enumerating all
solutions to the n-queens problem. It is representative of search
programs where most of the time is spent in backtracking search, and
we decided to test our approach against a wide variety of alternative
implementations to get a sense of the replay overhead on search-heavy
non-deterministic programs. Table \ref{table:nqueens} reports the
times for each implementation for different $n$.

Numbers in this section were collected using SML/NJ v110.82, OCaml 4.06, MLton
2018-02-07, the experimental OCaml-multicore runtime for OCaml 4.02.2
(dev0), and GNU Prolog 1.4.4, on a Lenovo T440s with an 1.6GHz Intel
Core i5 processor.

All implementations are exponential, and the slope
(the exponential coefficient) is essentially the same for all
implementations: for a given node in the computation tree, the
overhead of replaying the past is proportional to its depth, which is
bounded by $n$, and this bounded multiplicative slowdown factor becomes
a $log(n)$ additive overhead in log-scale.

The generic thermometer-based implementation (using Filinski's
construction with replay-based shift and reset) is noticeably slower
than the simpler replay-based implementation of non-determinism. On
the other hand, the optimizations described in
Section~\ref{sec:optimization} suffice to remove the additional
overhead: optimized thermometers are just as fast as
replay-nondeterminism.

We benchmarked Filinski's construction using the delimcc
library,\scite{kiselyov2010delimited} which is an implementation of
delimited control doing low-level stack copying; it is comparable in
speed to the generic thermometers. This is a nice result even though
delimcc is known not to be competitive with control operators in
runtimes designed to make them efficient.

We also wrote a version of the benchmark using the effect handlers
provided by the experimental Multicore OCaml
runtime.\scite{multicore-ocaml-effect-handlers}\footnote{This runtime
  offers very efficient one-shot continuations, but copying
  continuations is unsafe and may be less optimized.} The performance
is better than our replay-based non-determinism
(or optimized thermometers), but in the same logarithmic ballpark.

In contrast, in our SML benchmarks, Filinski's construction using the
native \texttt{call/cc} of SML/NJ is significantly faster than
replay-based approaches, closer to the indirect-style baseline. We
explain this by the fact that the SML/NJ compilation strategy and
runtime has made choices (call frames on the heap) that give a very
efficient \texttt{call/cc}, but add some overhead (compared to using
the native stack) to general computation. Indeed, in absolute time,
OCaml's indirect baseline is about 60$\%$ faster than the SML/NJ
baseline, and SML/NJ's \texttt{call/cc} approach is actually only
30$\%$ faster than OCaml's best replay-based implementations, and
slower than the Multicore-OCaml implementation.

To validate this hypothesis, we measured the performance of MLton on
the benchmarks that it supports (not thermometers, which require
a working \texttt{Unsafe.cast}; see
Section~\ref{subsec:universal-type}). The performance results for the
indirect and replay-nondeterminism version are extremely similar to
OCaml's, but the \texttt{call/cc} version is much slower than the one
of SML/NJ, similar to the performances of \texttt{delim/cc}; on MLton,
our replay-based technique is the best direct-style implementation of
non-determinism.

Finally, we wrote a Prolog version of N-queens (using Prolog's
direct-style backtracking search, not constraint solving). We were
surprised to find out that it is slower than all other
implementations. For N=13, GNU Prolog takes around 20s, while our
slowest ML versions run in 14s. Prolog may have efficient support for
backtracking search, but it seems disadvantaged by being
a dynamically-typed language without an aggressive JIT
implementation.

Our conclusion is that for computations dominated by nondeterministic
search, replay-based approaches are surprisingly practical: they offer
reasonable performances compared to other approaches to direct-style
nondeterminism that are considered practical.

\subsection{More benchmarks}

More benchmarks in different monads are included in Appendix \ref{app:parsing-benchmarks}.

\section{But Isn't This Impossible?}
\label{sec:compare-theoretical}

% TODO: either this section or the related work could contain
% a mention of a symmetric result, showing that exceptions *cannot* be
% macro-expressed using references and continuations alone:
%
%   Exceptions, Continuations and Macro-expressiveness
%   James Laird, ESOP 2002

In a 1990 paper, Matthias Felleisen presented formal notions of expressibility and macro-expressibility of one language feature in terms of others, along with proof techniques to show a feature cannot be expressed.\scite{felleisen1990expressive} Hayo Thielecke used these to show that exceptions and state together cannot macro-express continuations.\scite{thielecke2001contrasting} This is concerning, because, at first glance, this is exactly what we did.

First, a quick review of Felleisen's concepts: Expressibility and macro-expressibility help define what should be considered core to a language, and what is mere "syntactic sugar." An \emph{expression} is a translation from a language $\mathcal{L}$ containing a feature $\mathcal{F}$ to a language $\mathcal{L}^\prime$ without it which preserves program semantics. A key restriction is that an expression may only rewrite AST nodes that implement $\mathcal{F}$ and the descendants of these nodes. So, an expression of state may only rewrite assignments, dereferences, and expressions that allocate reference cells. Whle there is a whole-program transformation that turns stateful programs into pure programs, namely by threading a "state" variable throughout the entire program, this transformation is not an expression. A macro-expression is an expression which may rewrite nodes from $\mathcal{F}$, but may only move or copy the children of such nodes (technically speaking, it must be a term homomorphism). A classic example of a macro-expression is implementing the \code{+=} operator in terms of normal assignment and addition. A classic example of an expression which is not a macro-expression is desugaring \code{for}-loops into \code{while}-loops (it must dig into the loop body and modify every \code{continue} statement). Another one is implementing Java \code{finally} blocks (which need to execute an action after every \code{return} statement).

It turns out that Thielecke's proof does not immediately apply. First, it concerns \code{call/cc}-style continuations rather than delimited continuations. This seems like a minor limitation, since delimited continuations can be use to implement \code{call/cc}. But the second reason is more fundamental.

Thielecke's proof is based on showing that, in a language with exceptions and state but not continuations, all expressions of the following form with different $j$ are operationally equivalent:

\begin{equation*}
R_j = \lambda f.((\lambda x.\lambda y.(f\; 0;\; x:=\; !y; y:=j; !x)) (\text{ref}\; 0) (\text{ref}\; 0))
\end{equation*}

The intuition behind this equivalence is that the two reference cells are allocated locally and then discarded, and so the value stored in them can never be observed. However, with continuations, on the other hand, $f$ could cause the two assignments to run twice on the same reference cells.

This example breaks down because it cannot be expressed in our monadic reflection framework as is. The monadic reflection framework assumes there are no other effects within the program other than the ones implemented via monadic reflection. To write the $R_j$ using thermometer continuations and monadic reflection, the uses of \code{ref} must be changed from the native SML version to one implemented using the state monad. Then, when the computation is replayed, repeated calls to \code{ref} may return the same reference cell, allowing the state to escape, thereby allowing different $R_j$ to be distinguished. What this means is, because thermometer continuations require rewriting uses of mutable state, and not only uses of \code{shift} and \code{reset}, they are not a macro-expression. So, Thielecke's theorem does not apply.

\section{Related Work}

Our work is most heavily based on Filinski's work expressing monads using delimited control.\scite{Filinski94} We have also discussed theoretical results regarding the inter-expressibility of exceptions and continuations in Section \ref{sec:compare-theoretical}. Other work on implementing continuations using exceptions relate the two from a runtime-engineering perspective and from a typing perspective.

\paragraph{Continuations from stack inspection}

Oleg Kiselyov's \emph{delimcc} library\scite{kiselyov2010delimited} provides an implementation of delimited control for OCaml, based on the insight that the stack-unwinding facilities used to implement exceptions are also useful in implementing delimited control. Unlike our approach, \emph{delimcc} works by tightly integrating with the OCaml runtime, exposing low-level details of its virtual machine to user code. Its implementation relies on copying portions of the stack into a data structure, repurposing functionality used for recovering from stack overflows. It hence would not work for e.g.: many implementations of Java, which recover from stack overflows by simply deleting activation records. On the other hand, its low-level implementation makes it efficient and lets it persist delimited continuations to disk. A similar insight is used by Pettyjohn et al\scite{pettyjohn2005continuations} to implement continuations using a global program transformation.

\paragraph{Replay-based web programming} WASH Server Pages\scite{DBLP:conf/flops/Thiemann06} is a web framework that uses replay to persist state across multiple requests. In WASH, a single function can both display an HTML form to the end user, and take an action based on the user's response: the function is run twice, with the user's response stored in a replay log. MFlow ,\scite{mflow} another web framework, works along similar principles. Our work shows how the replay logs used by these frameworks are equivalent to serializing a continuation, revealing the connection between MFlow/WASH and continuation-based web frameworks.

\paragraph{Typing power of exceptions vs. continuations}
Lillibridge\scite{lillibridge1999unchecked} shows that exceptions introduce a typing loophole that can be used to implement unbounded loops in otherwise strongly-normalizing languages, while continuations cannot do this, giving the slogan "Exceptions are strictly more powerful than call/cc." As noted by other authors,\scite{thielecke2001contrasting} this argument only concerns the typing of exceptions rather than their execution semantics, and is inapplicable in languages that already have recursion.

\section{Conclusion}

Filinski's original construction of monadic reflection from delimited continuations, and delimited continuations from normal continuations plus state, provided a new way to program for the small fraction of languages which support first-class continuations. With our demonstration that exceptions and state are sufficient, this capability is extended to a large number of popular languages, including 9 of the TIOBE 10\scite{tiobe10} (all but C). While languages like Haskell with syntactic support for monads may not benefit from this construction, bringing advanced monadic effects to more languages paves the way for ideas spawned in the functional programming community to influence a broader population.

In fact, the roots of this paper came from an attempt to make one of the benefits of monads more accessible. We built a framework for Java where a user could write something that looks like a normal interpreter for a language, but, executed differently, it would become a flow-graph generator, a static analyzer, a compiler, etc. Darais\scite{darais2017abstracting} showed that this could be done by writing an interpreter in the monadic style (concrete interpreters run programs directly; abstract interpreters run them nondeterministically). We discovered this concurrently with Darais, and then discovered replay-based nondeterminism so that Java programmers could write normal, non-monadic programs.

%FIXME: It doesn't really provide hope
Despite the apparent inefficiency of thermometer continuations, the surprisingly good performance results of Section~\ref{sec:benchmarks}, combined with the oft-unused speed of modern machines, provide hope that the ideas of this paper can find their way into practical applications. Indeed, Filinski's construction is actually known as a way to make programs \emph{faster}. \scite{hinze2012kan}

%NOTE: Is this too self-aggrandizing?
Overall, we view finding a way to bring delimited control into mainstream languages as a significant achievement. We hope to see a flourishing of work with advanced effects now that they can be used by more programmers.

Working code for all examples and benchmarks, as well as the CPS-bind and direct-return optimizations, is available from \url{https://github.com/jkoppel/thermometer-continuations} .

\appendix

\section{Is this possibly correct?}

\label{sec:correctness}

In this section, we provide a correctness proof of our construction in
the simplest possible case: the replay-based implementation of
non-deterministic computations with a two-argument \code{choose x y}
effect, in a simplistic toy language.

\newcommand{\gramdef}{\mathrel{\mathord{:}\mathord{:=}}}
\newcommand{\letin}[3]{\mathsf{let}~#1\,=\,#2~\mathsf{in}~#3}
\newcommand{\amb}[2]{\mathsf{choose}~#1~#2}
\newcommand{\su}[1]{\mathsf{S}~#1}
\newcommand{\eqdef}{\stackrel{\mathsf{def}}{=}}

We consider programs with the following grammar. Notice that, for
simplicity, we only allow our choice operator $\amb x y$ to take
variables as parameters, rather than arbitrary expressions.
\[
  \begin{array}{l@{~}r@{~}l@{\quad}l}
  t,u & \gramdef &            & \mbox{terms} \\
    & \mid & x, y, z...       & \mbox{variables} \\
    & \mid & n \in \mathbb{N} & \mbox{number constant} \\
    & \mid & \su t            & \mbox{successor} \\
    & \mid & \letin x t u     & \mbox{let-binding} \\
    & \mid & \amb x y         & \mbox{choice} \\
  \end{array}
\]

\newcommand{\cont}[1]{\textcolor{blue}{#1}}
\newcommand{\term}[1]{\textcolor{black}{#1}}

\newcommand{\tsu}[1]{\term{\su{\term{#1}}}}
\newcommand{\csu}[1]{\cont{\su{\cont{#1}}}}
\newcommand{\tletin}[3]{\term{\letin{\term{#1}}{\term{#2}}{\term{#3}}}}
\newcommand{\cletin}[3]{\cont{\letin{\term{#1}}{\square}{(\term{#3},\cont{#2})}}}

\newcommand{\halt}{\cont{\mathsf{halt}}}

\newcommand{\mcont}[4]{(#1, #2, #3, #4)}
\newcommand{\mcontinit}[1]{\mcont {#1} \halt \emptyset \emptyset}

\noindent To make the execution semantics of non-deterministic choice precise,
we define a notion of abstract machine, the \emph{continuation
  machines}. Continuation machines are quadruples
$\mcont t {\cont{K}} {\cont{s}} R$ of a term, a current \emph{machine
  continuation} $\cont{K}$, a \emph{soup of threads} $\cont{s}$ that
are waiting to execute next, and a \emph{result} $R$, which is a list
of integers. Machine continuations are given by the following very
simple grammar, and thread soups $\cont{s}$ are lists of pairs
$(t, \cont{K})$ of a term and a continuation.
\[
  \begin{array}{l@{~}r@{~}l@{\quad}l}
  \cont{K} & \gramdef &            & \mbox{machine stacks} \\
    & \mid & \csu K                & \mbox{successor} \\
    & \mid & \cletin x K t         & \mbox{binding} \\
    & \mid & \halt                 & \mbox{halt instruction} \\
  \end{array}
  \qquad
  \begin{array}{l}
    \cont{s} \quad\gramdef\quad \emptyset \;\mid\; (t, \cont{K}) . \cont{s} \\
    R \quad\gramdef\quad \emptyset \;\mid\; n . R \\
  \end{array}
\]

\newcommand{\contrew}{\mathrel{\cont{\rightarrow}}}
\newcommand{\subst}[3]{#1[#2\leftarrow#3]}

\noindent We can now define a reduction relation
$\mcont t {\cont K} {\cont{s'}} R
\contrew
\mcont {t'} {\cont{K'}} {\cont{s'}} {R'}$ to
formally define the execution of those non-deterministic programs.  If
$\mcontinit t$ eventually reduces to
$\mcont n \halt \emptyset R$ for $n \in \mathbb{N}$, then the meaning of
the program is the result $n.R$.
\begin{mathpar}
  \begin{array}{l@{\quad\contrew\quad}l}
    \mcont {\tsu t} {\cont K} {\cont s} R & \mcont t {\csu K} {\cont s} R \\
    \mcont n {\csu K} {\cont s} R & \mcont {n+1} {\cont K} {\cont s} R \\

    \mcont {\tletin x t u} {\cont K} {\cont s} R & \mcont t {\cletin x K u} {\cont s} R \\
    \mcont n {\cletin x K u} {\cont s} R & \mcont {\subst u x n} {\cont K} {\cont s} R \\
  \end{array}

  \begin{array}{l@{\quad\contrew\quad}l}
    \mcont {\amb {n_1} {n_2}} {\cont K} {\cont s} R
    & \mcont {n_1} {\cont K} {(n_2, \cont{K}).\cont s} R \\

    \mcont n \halt {(n', \cont K).\cont s} R
    & \mcont {n'} {\cont K} {\cont s} {n.R} \\
  \end{array}
\end{mathpar}
The two last rules give the semantics of parallelism:
$\amb {n_1} {n_2}$ returns $n_1$, but adds the thread
$(n_2, \cont{K})$ to the soup, to be executed later, and the
$\halt$ rule adds a number to the result set and restarts the
next thread in the soup. Note that the rule for $\amb {n_1} {n_2}$
duplicates the continuation $\cont K$, pushing a copy of it on the
thread soup: this machine relies on a runtime that can duplicate
continuations.

\newcommand{\hist}[1]{\textcolor{red}{#1}}
\newcommand{\histrew}{\mathrel{\hist{\to}}}
\newcommand{\nexthist}[1]{#1 \hist{+ 1}}

\newcommand{\mhist}[6]{(#1, #2, #3, #4, #5)_{#6}}
\newcommand{\mhistinit}[1]{\mhist {#1} \halt \emptyset \emptyset \emptyset {#1}}

Our correctness argument comes from defining an alternative notion of
abstract machines for these programs, the \emph{history machines},
that model our replay-based technique. History machines are of the
form $\mhist t {\cont K} {\hist P} {\hist F} R u$, replacing the soup
$\cont s$ with a \emph{past} $\hist P$ and a \emph{future}
$\hist F$ which are just ordered lists of choices $\hist i$ --
numbers in $\{1, 2\}$. Finally, $u$ is the initial computation that
gets replayed; it is constant over the whole execution. We also define
an operation $\nexthist P$ corresponding to the \code{next_path}
function of Section~\ref{subsec:list-many}, which gets stuck on an
empty path: $\nexthist \emptyset$ is just $\nexthist \emptyset$.
\begin{mathpar}
  \begin{array}{l@{\quad\gramdef\quad}l}
    \hist{i} & 1 \mid 2 \\
  \end{array}

  \begin{array}{l@{\quad\gramdef\quad}l}
    \hist{P} & \emptyset \mid \hist{P}.\hist{i} \\
    \hist{F} & \emptyset \mid \hist{i}.\hist{F} \\
  \end{array}

  \begin{array}{l@{\quad\eqdef\quad}l}
    \nexthist{\hist{P}.\hist 1} & \hist{P}.2 \\
    \nexthist{\hist{P}.\hist 2} & \nexthist{\hist{P}} \\
  \end{array}
\\
  \begin{array}{l@{\;\histrew\;}l}
    \mhist {\tsu t} {\cont K} {\hist P} {\hist F} R u
    & \mhist t {\csu K} {\hist P} {\hist F} R u \\
    \mhist n {\csu K} {\hist P} {\hist F} R u
    & \mhist {n+1} {\cont K} {\hist P} {\hist F} R u \\

    \mhist {\tletin x t {t'}} {\cont K} {\hist P} {\hist F} R u
    & \mhist t {\cletin x K {t'}} {\hist P} {\hist F} R u \\
    \mhist n {\cletin x K {t'}} {\hist P} {\hist F} R u
    & \mhist {\subst {t'} x n} {\cont K} {\hist P} {\hist F} R u \\
  \end{array}

  \begin{array}{l@{\quad\histrew\quad}l}
    \mhist {\amb {n_1} {n_2}} {\cont K} {\hist P} \emptyset R u
    & \mhist {\amb {n_1} {n_2}} {\cont K} {\hist P} {1.\emptyset} R u \\

    \mhist {\amb {n_1} {n_2}} {\cont K} {\hist P} {(\hist i.\hist F)} R u
    & \mhist {n_{\hist i}} {\cont K} {(\hist P.\hist i)} {\hist F} R u \\

    \mhist n \halt {\hist P} \emptyset R u
    & \mhist u \halt \emptyset {\nexthist{\hist P}} {n.R} u \\
  \end{array}
\end{mathpar}
When we reach an instruction $\amb{n_1}{n_2}$, the machine picks the
value $n_{\hist i}$ (for $\hist i \in \{1, 2\}$) determined by its
future $\hist{F}$. When it finishes computing a value $n$ with past
history $\hist{P}$, it restarts the whole computation $u$ with future
$\nexthist{\hist{P}}$, exactly like our implementation. The semantics
of a term $t$ for these machines is given by the result $R$ such that
$\mhistinit t$ reduces to
$\mhist t \halt \emptyset {\nexthist\emptyset} R t$.

Remarkably, this machine never duplicates a continuation $\cont K$;
it can be implemented in language lacking this runtime capability. It
is also interesting to compare how $\cont K$ and
$(\hist P, \hist F)$ grow over time: $\cont K$ grows at each
reduction, including of non-effectful operations, while
$(\hist P, \hist F)$ only changes at the call sites of
$\amb {n_1} {n_2}$ and on $\halt$ -- it only records the effects in
the computation tree. Tracking the value of $\cont K$ in user code
would require inserting code around \emph{all}
operations\footnote{This typically requires compiler support. It can
  also be done with rich enough a macro system -- which we don't
  assume or require here. Indeed, you can use macros to redefine each
  operation of the language to first record the operation in a log and
  then perform the operation; a continuation can then be represented
  as the relevant fragment of the log. This technique has been
  mentioned to us, in private communication, by Matthias Felleisen.},
while $(\hist P, \hist F)$ can be tracked by
clever implementations of the effectful operators only.

\newcommand{\combrew}{\to}
\newcommand{\mcomb}[7]{(#1, {#2}_{#3}, #4, #5, #6)_{#7}}
\newcommand{\mcombinit}[1]{\mcomb {#1} \halt \emptyset \emptyset \emptyset \emptyset {#1}}

Correctness of the history machine means that the second reduction
semantics produces the same results as the first one. One way to do
this is to convince ourselves that the reductions of both machine
embed in the following \emph{combined} machines, where machine
continuations $\cont K$ are annotated with a past $\hist P$ --
including in the soup -- and the machine also holds a future
$\hist F$:

\begin{mathpar}
  \begin{array}{l@{\;\to\;}l}
    \mcomb {\tsu t} {\cont K} {\hist P} {\hist F} {\cont s} R u
    & \mcomb t {(\csu K)} {\hist P} {\hist F} {\cont s} R u \\
    \mcomb n {(\csu K)} {\hist P} {\hist F} {\cont s} R u
    & \mcomb {n+1} {\cont K} {\hist P} {\hist F} {\cont s} R u \\

    \mcomb {\tletin x t {t'}} {\cont K} {\hist P} {\hist F} {\cont s} R u
    & \mcomb t {(\cletin x K {t'})} {\hist P} {\hist F} {\cont s} R u \\
    \mcomb n {(\cletin x K {t'})} {\hist P} {\hist F} {\cont s} R u
    & \mcomb {\subst {t'} x n} {\cont K} {\hist P} {\hist F} {\cont s} R u \\
  \end{array}

  \begin{array}{l@{\quad\to\quad}l}
    \mcomb {\amb {n_1} {n_2}} {\cont K} {\hist P} \emptyset {\cont s} R u
    & \mcomb {n_1} {\cont K} {\hist P . \hist 1} \emptyset
        {(n_2, \cont{K}_{\hist P . \hist 2}).\cont{s}} R u \\

    \mcomb n \halt {\hist P} \emptyset {(n', \cont{K}_{\hist{P'}}).\cont{s}} R u
    & \mcomb {n'} {\cont K} {\hist{P'}} \emptyset {\cont s} {n.R} u \\

    \mcomb {\amb {n_1} {n_2}} {\cont K} {\hist P} {\hist i . \hist F} {\cont s} R u
    & \mcomb{n_{\hist i}} {\cont K} {\hist P . \hist i} {\hist F} {\cont s} R u \\
  \end{array}
\end{mathpar}
For the easy rules (the first four), if you erase histories you get
valid continuation reductions, and if you erase soups you get valid
history reductions -- both machines coincide exactly on those
reductions. The subtleties come from the hard rules, the second pack.

The first hard rule is a valid continuation rule, and corresponds to
the composition of two valid history rules. The second hard rule is
a valid continuation rule, but in the history system it corresponds to
the $\halt$ rule followed by an arbitrarily-long series of replay
computations, which use the third hard rule (they are the only part of
the simulation argument to use this rule). More precisely, we claim
that the following sequence, \emph{after} erasing into history
machines, is a valid sequence of history reductions.
\footnote{As given, it is not a sequence of reductions for
  the combined machine, but this could be obtained by changing the
  reduction rule for
  $\mcomb n \halt {\hist P} \emptyset {(n', \cont{K}_{\hist{P'}}).\cont{s}} R u$
  to reduce to
  $\mcomb u \halt \emptyset {\hist P'} {\cont s} {n.R} u$
  instead of
  $\mcomb {n'} {\cont K} {\hist{P'}} \emptyset {\cont s} {n.R} u$.
  The fact that these two reduction choices are equivalent, from the
  point of view of the final result, is precisely the Replay
  Theorem~\ref{thm:replay}.}

\begin{mathpar}
   \mcomb n \halt {\hist P} \emptyset {(n', \cont{K}_{\hist{P'}}).\cont{s}} R u

   \histrew

   \mcomb u \halt \emptyset {\hist P'} {\cont s} {n.R} u

   \histrew^*

   \mcomb {n'} {\cont K} {\hist{P'}} \emptyset {\cont s} {n.R} u
\end{mathpar}

\noindent We now establish the technical results that show that this is indeed
a valid history reduction sequence. The first step looks like the
third history reduction rule, except that a seemingly-arbitrary past
$\hist{P'}$ taken from the soup replaces $\nexthist{\hist{P}}$. We
show that a $\hist{P'}$ at the head of the soup must in fact be equal
to $\nexthist{\hist{P}}$ -- this is part of the Timeline Invariant
Lemma~\ref{lem:timeline}. The second step corresponds to replaying the
computation that reached $(n', \cont{K}_{\hist{P'}})$; its correctness
is established by the Replay Theorem~\ref{thm:replay}.

\begin{lemma}[Timeline invariant]
\label{lem:timeline}
  For any combined machine obtained by reducing
  a valid initial configuration, the first history at the top of the
  soup is the successor of the continuation's history, and the second
  history in the soup is the successor of the first, etc.
  Formally:
  \begin{enumerate}
  \item For any configuration of the form
    $\mcomb n {\cont K} {\hist P} \emptyset {(n', \cont{K'}_{\hist{P'}}).\cont{s}} R u$,
    we have $\hist{P'} = \nexthist{\hist P}$.
  \item For any (sub)-soup of the form
    $(n,\cont{K}_{\hist P}).(n',\cont{K'}_{\hist{P'}}).\cont{s}$
    occurring inside a configuration, we have
    $\hist{P'} = \nexthist{\hist P}$.
  \end{enumerate}
\end{lemma}

\begin{proof}
  All combined reductions preserve this invariant. The only
  non-trivial case being the first hard rule, which adds a new
  history-annotated continuation to the soup.
  \begin{mathpar}
    \mcomb {\amb {n_1} {n_2}} {\cont K} {\hist P} \emptyset {\cont s} R u

    \to

    \mcomb {n_1} {\cont K} {\hist P . \hist 1} \emptyset
      {(n_2, \cont{K}_{\hist P . 2}).\cont{s}} R u
  \end{mathpar}
  By assumption on the input configuration, the first history on top
  of the soup ${\cont s}$ (if any) is $\nexthist{\hist P}$. This
  is also equal to $\nexthist{(\hist P.\hist 2)}$, so the new
  soup $(n_2, \cont{K}_{\hist P . \hist 2}).\cont{s}$ respects
  the invariant (2). Furthermore, the new history on the top of the
  soup, $\hist P . \hist 2$, is indeed the successor of the
  history of the new annotated continuation,
  $\cont{K}_{\hist P . \hist 1}$, so the invariant (1) is also respected.
\end{proof}

\newcommand{\topure}{\to_{\mathsf{pure}}}
\newcommand{\init}{\mathsf{init}}
\newcommand{\replay}{\mathsf{replay}}

\begin{definition}[Pure reduction]
  We define the \emph{pure} reductions as the subset of combined
  machine reductions that do not depend on the soup $\cont{s}$ or
  result $R$ -- the four easy rules, plus the last hard rule. We write
  $(\topure)$ for pure reductions.
\end{definition}

Notice that pure reduction sequences are all of the following form
\footnote{We write $\hist P . \hist F$ and $\hist P . \hist P'$ and $\hist F . \hist F'$ for the concatenation of histories.}
\begin{mathpar}
  \mcomb {t_1} {(\cont{K_1})} {\hist P} {\hist{P'}.\hist{F}} {\cont s} R u

  \topure^*

  \mcomb {t_2} {(\cont{K_2})} {\hist{P}.\hist{P'}} {\hist F} {\cont s} R u
\end{mathpar}
\noindent The soup and results are unchanged, and a fragment of history moves
from the future to the past.

\begin{lemma}[Monotonicity]
  \label{lem:monotonicity}
  Pure reductions are insensitive to changes of soups and results,
  and they are preserved by appending to the future.
  If
  $\mcomb {t_1} {\cont{K_1}} {\hist P} {(\hist{P'}.\hist{F})} {\cont s} R u
  \topure^*
  \mcomb {t_2}  {\cont{K_2}} {(\hist{P}.\hist{P'})} {\hist F} {\cont s} R u$,
  then for any $\cont{s'}$, $R'$ and $\hist{F'}$ we also have
  $\mcomb {t_1} {\cont{K_1}}
    {\hist P} {(\hist{P'}.\hist{F}.\hist{F'})} {\cont{s'}} {R'} u
  \topure^*
  \mcomb {t_2} {\cont{K_2}}
    {(\hist{P}.\hist{P'})} {(\hist{F}.\hist{F'})} {\cont{s'}}  {R'} u$.
\end{lemma}

\begin{definition}[Replay configuration $\replay(c)$]
  We define the \emph{replay configuration} of an arbitrary configuration by
  $
    \replay{\mcomb t {\cont K} {\hist P} {\hist F} {\cont s} R u}
    \eqdef
    \mcomb u \halt \emptyset {(\hist{P}.\hist{F})} {\cont s} R u
  $.
\end{definition}

\begin{theorem}[Replay]
\label{thm:replay}
  If a combined configuration $c$ is reachable from an initial
  configuration, then its replay configuration $\replay(c)$
  reduces back to $c$ using only pure steps:
  \begin{mathpar}
    (\mcombinit u \combrew^* c) \quad\implies\quad (\replay(c) \topure^* c))
  \end{mathpar}
\end{theorem}

\begin{proof}
  The proof proceeds by induction on the combined reduction sequence from an
  initial state $\mcombinit u$ to an arbitrary configuration
  $\mcomb n {\cont K} {\hist P} {\hist F} {\cont s} R u$, building
  a corresponding sequence of combined reduction steps starting from its replay
  configuration $\mcomb u \halt \emptyset {(\hist{P}.\hist{F})} {\cont s} R u$.

  Either the reduction sequence is empty, in which case the result is
  immediate, or we consider its last reduction step
  $c' \combrew c$. By induction hypothesis, $c'$ is purely reachable from $\replay(c')$,
  and we have to transport this reduction $replay(c') \topure^* c'$ into
  a reduction $\replay(c) \topure^* c$.

  We reason by case analysis on the reduction $c' \combrew c$.
  If it is a pure reduction, we have $\replay(c') = \replay(c)$ and
  can conclude by replaying the reduction from $c'$ unchanged. The
  delicate cases correspond to the two impure reductions.

  \paragraph{First impure reduction}
  \begin{mathpar}
    \mcomb {\amb {n_1} {n_2}} {\cont K} {\hist P} \emptyset {\cont s} R u

    \to

    \mcomb {n_1} {\cont K} {\hist P . \hist 1} \emptyset
      {(n_2, \cont{K}_{\hist P . \hist 2}).\cont{s}} R u
  \end{mathpar}
  Let us define $\cont{s'}$ as $(n_2, \cont{K}_{\hist P . 2}).\cont{s}$.
  We want to show that
  $\mcomb {n_1} {\cont K} {\hist P . \hist 1} \emptyset {\cont s'} R u$
  is purely reachable from
  $\mcomb u \halt \emptyset {\hist P . \hist 1} {\cont s'} R u$.
  By induction hypothesis, we know that
  $\mcomb u \halt \emptyset {\hist P} {\cont s} R u$
  purely reaches
  $\mcomb {\amb {n_1} {n_2}} {\cont K} {\hist P} \emptyset {\cont s} R u$
  .
  By monotonicity (Lemma~\ref{lem:monotonicity}), we know
  that $\mcomb u \halt \emptyset {\hist P . \hist 1} {\cont{s'}} R u$ purely
  reaches
  $\mcomb {\amb {n_1} {n_2}} {\cont K} {\hist P} {\hist 1 . \emptyset} {\cont{s'}} R u$,
  and we conclude with a pure reduction step.
  This proof can be represented better in diagrammatic form:

  \begin{mathpar}
  \begin{tikzcd}
    \mcombinit u
    \arrow[to, r, ""{name=Init, below}, "*" at end]
    &
    \mcomb {\amb {n_1} {n_2}} {\cont K} {\hist P} \emptyset {\cont s} R u
    \arrow[to, r]
    \arrow[equal, d]
    &
    \mcomb {n_1} {\cont K} {\hist P . \hist 1} \emptyset {\cont{s'}} R u
    \arrow[equal, dd]
    \\
    \mcomb u \halt \emptyset {\hist P} {\cont s} R u
    \arrow[to, r, "\mathsf{pure}"'{name=IndB}, ""{name=Ind}, "*" at end]
    &
    \mcomb {\amb {n_1} {n_2}} {\cont K} {\hist P} \emptyset {\cont s} R u
    &
    \\
    \mcomb u \halt \emptyset {\hist P . \hist 1 . \emptyset} {\cont{s'}} R u
    \arrow[to, r, "\mathsf{pure}"', ""{name=Mon}, "*" at end]
    &
    \mcomb {\amb {n_1} {n_2}} {\cont K} {\hist P} {\hist 1 . \emptyset} {\cont{s'}} R u
    \arrow[to, r, "\mathsf{pure}"']
    &
    \mcomb {n_1} {\cont K} {\hist P . \hist 1} \emptyset {\cont{s'}} R u
    \arrow[mapsto, ll, bend left=10, "\text{replay}"]

    \arrow[Rightarrow,from=Init,to=Ind,"\text{ind. hyp.}"]
    \arrow[Rightarrow,from=IndB,to=Mon,"\text{monotonicity}"]
  \end{tikzcd}
  \end{mathpar}

  \paragraph{Second impure reduction}
  \begin{mathpar}
    \mcomb {n_1} \halt {\hist P} \emptyset {(n_2, \cont{K}_{\hist{P'}}).\cont{s}} R u

    \to

    \mcomb {n_2} {\cont K} {\hist{P'}} \emptyset {\cont s} {n_1.R} u
  \end{mathpar}

\noindent   By the Timeline Invariant (Lemma~\ref{lem:timeline}), we know that
  $\hist{P'} = \nexthist{\hist{P}}$. Let us pose $R' \eqdef {n_1}.R$.

  In the reduction sequence from the initial configuration to
  $\mcomb {n_1} \halt {\hist P} \emptyset {(n_2, \cont{K}_{\hist{P'}}).\cont{s}} R u$,
  let us consider the first reduction that pushed the paused computation
  $(n_2, \cont{K}_{\hist{P'}})$ onto the soup $\cont{s}$. This reduction
  can only be of the form
  \begin{mathpar}
    \mcomb {\amb {n_1} {n_2}} {\cont K} {\hist {P_0}} \emptyset {\cont s} R u
    \to
    \mcomb {n_1} {\cont K} {\hist {P_0} . \hist 1} \emptyset
        {(n_2, \cont{K}_{\hist {P_0} . \hist 2}).\cont{s}} R u
  \end{mathpar}
  for a certain history $\hist{P_0}$ such that
  $\hist{P_0}.\hist 2 = \nexthist{\hist{P}}$.

  The proof follows from these definitions by the following diagram:
  \begin{mathpar}
  \begin{tikzcd}
    \mcombinit u
    \arrow[to, r, ""{name=Init, below}, "*" at end]
    &
    \mcomb {\amb{n_1}{n_2}} {\cont K} {\hist {P_0}} \emptyset {\cont s} R u
    \arrow[to, r, "*" at end]
    \arrow[equal,d]
    &
    \mcomb {n_2} {\cont K} {\nexthist{\hist{P}}} \emptyset {\cont s} {n_1.R} u
    \arrow[equal,dd]
    \\
    \mcomb u \halt \emptyset {\hist {P_0}} {\cont s} R u
    \arrow[to, r, "\mathsf{pure}"'{name=IndB}, ""{name=Ind}, "*" at end]
    &
    \mcomb {\amb{n_1}{n_2}} {\cont K} {\hist {P_0}} \emptyset {\cont s} R u
    &
    \arrow[Rightarrow,from=Init,to=Ind,"\text{ind. hyp.}"]
    \\
    \mcomb u \halt \emptyset {\hist{P_0}.\hist 2.\emptyset} {\cont s} {R'} u
    \arrow[to, r, "\mathsf{pure}"', ""{name=Mon}, "*" at end]
    &
    \mcomb {\amb{n_1}{n_2}} {\cont K}
      {\hist{P_0}} {\hist 2.\emptyset} {\cont s} {R'} u
    \arrow[to,r, "\mathsf{pure}"']
    &
    \mcomb {n_2} {\cont K} {\hist{P_0}.\hist 2} \emptyset {\cont s} {R'} u
    \arrow[mapsto, ll, bend left=10, "\text{replay}"]

    \arrow[Rightarrow,from=Init,to=Ind,"\text{ind. hyp.}"]
    \arrow[Rightarrow,from=IndB,to=Mon,"\text{monotonicity}"]
  \end{tikzcd}
  \end{mathpar}
\end{proof}

\section{Parsing benchmarks}
\label{app:parsing-benchmarks}

To get a sense of the performance cost of thermometer continuations in
realistic effectful programs, we wrote direct-style versions of
monadic parsing programs.

There are three benchmarks, all implemented in SML: \textsc{INTPARSE-GLOB}, \textsc{INTPARSE-LOCAL}, and \textsc{MONADIC-ARITH-PARSE}. They use two different monads. Depending on the monad, we gave three to six implementations of each benchmark. Each \textbf{Indirect} implementation implements the program pure-functionally, without monadic effects. The \textbf{ThermoCont} and \textbf{Filinski} implementations use monadic reflection, implemented via thermometer continuations and Filinski's construction, respectively (using SML's native \texttt{call/cc} support). For the nondeterminism and failure monads, our optimizations apply, given in \textbf{Opt. ThermoCont}.

Numbers were collected using SML/NJ v110.80.\scite{appel1991standard} We could not port our implementation to MLTON,\scite{weeks2006whole} the whole-program optimizing compiler for SML, because it lacks \code{Unsafe.cast}. The experiments in this section only were conducted on a 2015 MacBook Pro with a 2.8 GHz Intel Core i7 processor. All times shown are the average of 5 trials, except for \textsc{MONADIC-ARITH-PARSE}, as discussed below.

The twin benchmarks \textsc{INTPARSE-GLOB} and \textsc{INTPARSE-LOCAL} both take a list of numbers as strings, parse each one, and return their sum. They both use a monadic failure effect (like Haskell's \code{Maybe} monad), and differ only in their treatment of  strings which are not valid numbers: \textsc{INTPARSE-GLOB} returns failure for the entire computation, while \textsc{INTPARSE-LOCAL} will recover from failure and ignore any malformed entry. Table \ref{table:parse-int-global} gives the results for \textsc{INTPARSE-GLOB}. For each input size $n$, we constructed both a list of $n$ valid integers, as well as one which contains an invalid string halfway through. Table \ref{table:parse-int-local} gives the results for \textsc{INTPARSE-LOCAL}. For each $n$, we constructed lists of $n$ strings where every $1/100$th, $1/10$th, or $1/2$nd string was not a valid int. For \textsc{INTPARSE-GLOB}, unoptimized thermometer continuations wins, as it avoids Filinski's cost of \code{call/cc}, the optimized version's reliance on closures, as well as the indirect approach's cost of wrapping and unwrapping results in an \code{option} type. For \textsc{INTPARSE-LOCAL}, unoptimized thermometer continuations lost out to the indirect implementation, as thermometer continuations here devolve into raising an exception for bad input, which is slower than wrapping a value in an option type.

Finally, benchmark \textsc{MONADIC-ARITH-PARSE} is a monadic parser in the style of Hutton and Meijer.\scite{hutton1998monadic} These programs input an arithmetic expression, and return the list of results of evaluating any {\it prefix} of the string which is itself a valid expression. The Filinski and ThermoCont implementations closely follow Hutton and Meijer, executing in a "parser" monad with both mutable state and nondeterminism (equivalent to Haskell's \code{StateT List} monad). The indirect implementation inlines the monad definition, passing around a list of (remaining input, parse result) pairs. We did not provide an implementation with optimized thermometer continations, as we have not yet found how to make our optimizations work with the state monad. Note that all three implementations use the same algorithm, which, although the code imeplementing it is beautiful (at least for the monadic versions), is exponentially slower than the standard LL/LR parsing algorithms.

Table \ref{table:arith-parse} reports the average running time of each implementation on 30 random arithmetic expressions with a fixed number of leaves. There was very high variance in the running time of different inputs, based on the ambiguity of prefixes of the expression. For inputs with 40 leaves, the fastest input took under 25ms for all implementations, while the slowest took approximately 25 minutes on the indirect implementation, 5 hours and 43 minutes for Filinski's construction, and 9 hours 50 minutes for thermometer continuations.

Overall, these benchmarks show that the optimizations of Section~\ref{sec:optimization} can provide a substantial speedup, and there are many computations for which thermometer continuations do not pose a prohibitive cost. Thermometer continuations are surprisingly competitive with Filinski's construction, even though SML/NJ is known for its efficient \code{call/cc}, and yet thermometer continuations can be used in far more programming languages.

\begin{table}
\begin{center}
\caption{Benchmark \textsc{INTPARSE-GLOB}}
\begin{tabular}{r|r|rrrrrr}
\hline
& \textbf{Bad input?} & \textbf{10,000} & \textbf{50,000} & \textbf{100,000} & \textbf{500,000} & \textbf{1,000,000} & \textbf{5,000,000} \\
 \hline
\multirow{2}{*}{\textbf{Indirect}} & Y & 0.001s & 0.004s & 0.026s & 0.159s & 0.260s & 36.260s \\
& N & 0.002s & 0.014s & 0.051s & 0.213s & 0.371s & 1m54.072s \\
\hline
\multirow{2}{*}{\textbf{Filinski}} & Y & 0.004s & 0.018s & 0.012s & 0.172s & 0.258s & 28.719s  \\
& N & 0.004s & 0.020s & 0.014s & 0.221s & 0.360s & 1m26.603s \\
\hline
\multirow{2}{*}{\textbf{Therm.}} & Y & 0.003s & 0.008s & 0.024s & 0.167s & 0.260s & 27.461s \\
& N & 0.003s & 0.011s & 0.029s & 0.223s & 0.367s & 1m20.841s \\
\hline
\multirow{2}{*}{\textbf{Therm. Opt}} & Y & 0.000s & 0.013s & 0.015s & 0.197s & 0.293s & 27.298s \\
& N & 0.002s & 0.015s & 0.024s & 0.247s & 0.364s & 1m23.433s \\
\hline
\end{tabular}
\label{table:parse-int-global}
\end{center}
\end{table}

\begin{table}
\begin{center}
\caption{Benchmark \textsc{INTPARSE-LOCAL}}
\begin{tabular}{r|r|rrrrrr}
\hline
& \textbf{\% bad input} & \textbf{10,000} & \textbf{50,000} & \textbf{100,000} & \textbf{500,000} & \textbf{1,000,000} & \textbf{5,000,000} \\
 \hline
\multirow{3}{*}{\textbf{Indirect}} & 1\% & 0.002s & 0.011s & 0.060s & 0.232s & 0.404s & 2m04.316s \\
& 10\% & 0.002s & 0.002s & 0.053s & 0.221s & 0.375s & 1m40.251s \\
& 50\% & 0.000s & 0.001s & 0.014s & 0.176s & 0.257s & 15.515s \\
\hline
\multirow{3}{*}{\textbf{Filinski}} & 1\% & 0.001s & 0.048s & 0.053s & 0.264s & 0.456s & 3m19.265s \\
& 10\% & 0.004s & 0.043s & 0.048s & 0.234s & 0.420s & 3m22.245s \\
& 50\% & 0.004s & 0.008s & 0.050s & 0.182s & 0.301s & 23.632s \\
\hline
\multirow{3}{*}{\textbf{Therm.}} & 1\% & 0.002s & 0.045s & 0.064s & 0.344s & 0.470s & 4m28.700s \\
& 10\% & 0.003s & 0.028s & 0.064s & 0.266s & 0.445s & 3m26.809s \\
& 50\% & 0.001s & 0.023s & 0.067s & 0.214s & 0.353s & 22.750s \\
\hline
\multirow{3}{*}{\textbf{Therm. Opt.}} & 1\% & 0.002s & 0.030s & 0.059s & 0.227s & 0.403s & 3m01.981s \\
& 10\% & 0.002s & 0.030s & 0.058s & 0.218s & 0.386s & 2m37.321s \\
& 50\% & 0.002s & 0.023s & 0.055s & 0.183s & 0.289s & 27.546s \\
\hline
\end{tabular}
\label{table:parse-int-local}
\end{center}
\end{table}

\begin{table}
\begin{center}
\caption{Benchmark \textsc{MONADIC-ARITH-PARSE}}
\begin{tabular}{r|rrrr}
\hline
 & \textbf{10} & \textbf{20} & \textbf{30} & \textbf{40} \\
 \hline
\textbf{Indirect} & 0.011s & 0.163s & 1.908s & 2m39.860s \\
\textbf{Filinski} & 0.116s & 1.638s & 19.035s & 25m18.794s \\
\textbf{Therm.} & 0.184s & 2.540s & 30.229s & 39m58.183s \\
\hline
\end{tabular}
\label{table:arith-parse}
\end{center}
\end{table}

\begin{acks}                            %% acks environment is optional
  We warmly thank Ben Lippmeier, Edward Z. Yang, Adam Chlipala, and the anonymous
  reviewers for the excellent feedback they provided on this article.

  Gerg{\"o} Barany helped benchmark our Prolog implementation; he
  suggested using GNU Prolog instead of SWI Prolog, which was giving
  much slower results.

  This material is based upon worked supported by the \grantsponsor{GS1122374}{National Science Foundation}{} under Grant No. 1122374. Any opinions, findings, and
  conclusions or recommendations expressed in this material are those
  of the authors and do not necessarily reflect the views of the
  National Science Foundation.
\end{acks}

%% Bibliography
\bibliography{cont.bib}
\end{document}